\documentclass[lettersize,journal]{IEEEtran}

\usepackage{colortbl}
\usepackage{bm}
\usepackage{balance}
\usepackage{graphicx}
\usepackage{epstopdf}
\usepackage{amsmath,amsfonts}
\usepackage{algorithmic}
\usepackage{algorithm}
\usepackage{array}
\usepackage[caption=false,font=normalsize,labelfont=sf,textfont=sf]{subfig}
\usepackage{textcomp}
\usepackage{stfloats}
\usepackage{url}
\usepackage{verbatim}
\usepackage{graphicx}
\usepackage{cite}
\usepackage{amsthm}
\usepackage{booktabs}
\newtheorem{theorem}{{Theorem}}
\newtheorem{lemma}{{Lemma}}
\newtheorem{definition}{{Definition}}

\newtheorem{remark}{{Remark}}
\newtheorem{assumption}{Assumption}
\begin{document}

\title{Data Heterogeneity-Aware Client Selection for Federated Learning in Wireless Networks}

\author{Yanbing Yang,
	Huiling~Zhu, Wenchi~Cheng, Jingqing~Wang, Changrun~Chen, and Jiangzhou~Wang,~\IEEEmembership{Fellow,~IEEE}
}



\maketitle

\begin{abstract}
Federated Learning (FL) enables mobile edge devices, functioning as clients, to collaboratively train a decentralized model while ensuring local data privacy. However, the efficiency of FL in wireless networks is limited not only by constraints on communication and computational resources but also by significant data heterogeneity among clients, particularly in large-scale networks. This paper first presents a theoretical analysis of the impact of client data heterogeneity on global model generalization error, which can result in repeated training cycles, increased energy consumption, and prolonged latency. Based on the theoretical insights, an optimization problem is formulated to jointly minimize learning latency and energy consumption while constraining generalization error. A joint client selection and resource allocation (CSRA) approach is then proposed, employing a series of convex optimization and relaxation techniques. Extensive simulation results demonstrate that the proposed CSRA scheme yields higher test accuracy, reduced learning latency, and lower energy consumption compared to baseline methods that do not account for data heterogeneity.

\end{abstract}


\section{Introduction} \label{introduction}
\IEEEPARstart{T}{he} rapid proliferation of mobile edge devices has driven the expansion of artificial intelligence (AI)-powered applications such as autonomous driving and smart cities \cite{9369324}. 
Supported by AI, these applications can leverage large volumes of training data, such as images and videos, to enhance the adaptability of the model across diverse scenarios and users \cite{arivazhagan2019federated}. However, collecting sufficient training data from multiple edge devices can impose significant transmission overhead and expose raw data to the public \cite{tian2022fedbert}. Federated learning (FL) presents a promising distributed and collaborative solution to this challenge by enabling mobile edge devices to train a shared model without the need to exchange raw data \cite{mcmahan2017communication}. 
With its collaborative framework, FL allows lightweight devices to acquire knowledge from distributed data silos, thereby reducing centralized data dependence and enhancing the generalization ability of the global model compared with isolated local training \cite{nguyen2021federated}. 
Consequently, FL alleviates the computation burden on the central server and offers a scalable approach for resource-limited networks. Nevertheless, several challenges remain in deploying FL over wireless networks. \textit{1) Training cost}. Frequent local model uploads combined with stringent latency requirements incur substantial communication overhead and intensive utilization of computational resources, resulting in substantial energy consumption. Moreover, the naive approach in which all clients participate in every training round introduces excessive training overhead, especially in large-scale FL deployments \cite{9237168}.
\textit{2) Data heterogeneity}. The spatio-temporal variability of client locations and personalized usage behaviors leads to heterogeneous data distributions across clients. Such heterogeneity induces a discrepancy between local and global optimization objectives, which in turn causes significant variations in model update directions across training rounds. In extreme cases, the final global model may become biased toward specific local datasets, thereby degrading its performance on the population dataset and even leading to model divergence \cite{zhao2018federated}.

Various strategies have been proposed to reduce training costs in FL. One promising approach is the selective client participation, which has been shown to be more efficient than full client participation. By targeting the most suitable clients, this method not only reduces computational workload but also enhances radio resource utilization, particularly under communication constraints such as bandwidth limitations, energy budgets, and latency-sensitive requirements \cite{9237168, vehicle_selection_fl, 9264742, 9712615, 9709639, 9210812, 9261995}.  
However, these approaches largely neglect the effects of data heterogeneity among clients. Selecting clients solely based on channel conditions tends to bias the global model update toward the local datasets of well-connected clients, causing a divergence from the true population-level gradient descent direction of the global model. Consequently, the learned model may suffer performance degradation on the population dataset, reducing generalization capability. Additionally, under highly skewed data distributions, such bias can significantly increase the number of training rounds needed to achieve target accuracy, leading to higher training costs.  

Several approaches attempted to address the adverse effects of data heterogeneity on model performance \cite{ 10375295, 9528995, 10192896, sattler2020clustered, 10241292, 10876775, 9882362, 10818753, 9741255,11037111,chen2024heterogeneity}. 
In \cite{10375295, 9528995, 10192896, sattler2020clustered, 10241292}, data heterogeneity is defined as the similarity of local and global loss gradients, or the similarity among local gradients across clients. Specifically, \cite{10375295} aims to accelerate FL convergence by selecting clients whose local gradients align with the global gradient. By adjusting aggregation weights according to gradient discrepancies, \cite{9528995} targets to correct the model bias induced by non-independent and identically distributed (non-IID) data.
Methods in \cite{sattler2020clustered, 10192896, 10241292} further cluster clients based on gradient similarity and train a separate model for each cluster, implicitly assuming that clients within the same cluster share similar data distributions. However, these gradient-based approaches are founded on the implicit assumption that gradient discrepancies can accurately represent data distribution heterogeneity, lacking a rigorous theoretical analysis of this connection. Moreover, the computation of gradient similarity incurs substantial calculation costs, which limits its practicality in large-scale FL networks. 
In contrast, studies in \cite{10876775, 9882362, 10818753, 9741255, 11037111,chen2024heterogeneity} explicitly define data heterogeneity in terms of the divergence between local and global data distributions. These studies aim to mitigate the effects of data divergence by aggregating datasets or by selectively choosing clients whose local distributions are more aligned with the global distribution. While they demonstrate the advantages of client selection based on data distribution, there is a lack of robust guarantees concerning global model generalization error, which can result in repeated training cycles, ultimately leading to increased energy consumption and prolonged latency in FL. Moreover, the existing methods \cite{10375295, 10876775, 9882362,10818753, sattler2020clustered, 10192896, 10241292, 9741255, 9528995, 11037111,chen2024heterogeneity} primarily focus on mitigating the impact of non-IID data on the convergence process, by minimizing empirical loss for the clients participating in each training round. However, the population loss, representing the model's actual performance across all potential data, remains largely unaddressed. Optimizing empirical loss solely for the participating clients does not ensure that the model performs well on the broader dataset. Consequently, ensuring the effectiveness and generalization capability of the global model in FL continues to be a significant challenge. Furthermore, the majority of these methods neglect channel variations in wireless networks, which further impede global model performance. Clients differ in dataset size, data distribution, and communication channel conditions. Therefore, considering any single factor in isolation is insufficient to ensure FL effectiveness in wireless networks, which depends on the interplay between training cost and global model generalization capability.


Motivated by the challenges identified in FL within wireless networks, this paper aims to reduce training costs jointly associated with learning latency and energy consumption, while considering the impact of data heterogeneity on global model performance, which is represented by the global model generalization error. Specifically, we characterized client data heterogeneity and theoretically analyzed the upper bound of generalization error by decomposing it into a generalization-related error and a training-related error. Based on the analytical results, we formulated the interplay between FL training cost and global model generalization capability as a training cost minimization problem, imposing constraints on generalization error for each training round. To address the complexity of the original problem, we decompose it into tractable subproblems and propose a dedicated algorithm to solve it effectively. The main contributions of this paper are summarized as follows:
%
%
%
\begin{itemize}
	\item[$\bullet$] 
	We provided the theoretical analysis of the impact of data heterogeneity on FL model performance, revealing that data heterogeneity, quantified by the KL divergence between local and global data distributions, interacts with the algorithmic stability of FL, influencing generalization performance.
	Furthermore, we proved that when the FL algorithm is uniformly stable, reducing distribution divergence among clients and increasing the total dataset size can significantly tighten the generalization error bounds.
	\item[$\bullet$] By imposing the FL global model performance as a constraint, we formulated a training cost minimization problem, and decomposed it into two sub-problems, recognizing it as a mixed-integer non-linear problem (MINLP). We then proposed a joint client selection and resource allocation approach (CSRA) to efficiently solve the subproblems with low computational complexity. 
	\item[$\bullet$] 
	Extensive simulations were conducted to evaluate the effects of data distribution divergence and dataset size on FL performance.  The results demonstrated that the proposed CSRA scheme significantly enhances model performance while reducing latency and energy consumption compared to baseline methods.	
\end{itemize}

The remainder of this paper is organized as follows: Section \ref{sec:System_model} presents the system model, data heterogeneity, and training cost. Section \ref{sec:Impact of data heterogeneity} analyzes the impact of data heterogeneity on FL. The training cost optimization problem is formulated, and a joint client selection and resource allocation method is proposed in Section \ref{problem}. 
Section \ref{simulation} discusses the simulations and results. Finally, Section \ref{conclusions} concludes the paper.
\begin{figure}[tbp]
	\centering
	\vspace{-8pt}
	\includegraphics[scale=0.63]{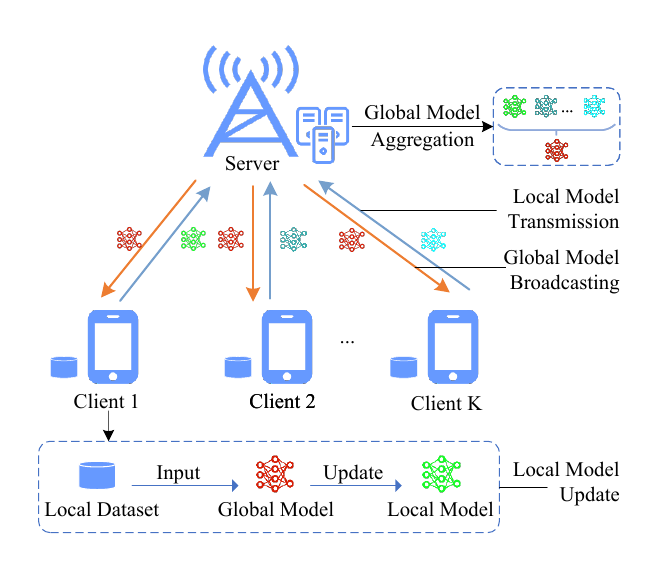}
	\caption{The architecture of federated learning over a wireless network.} \label{fig:pfl_system_model}
	\vspace{-10pt}
\end{figure}
\section{System Model}\label{sec:System_model}
As shown in Fig.~\ref{fig:pfl_system_model}, we consider a network including $K$ clients, represented by the set $\mathcal{K} = \{1, \cdots, k, \cdots, K\}$, which are connected to a central server via wireless channels. Each client $k$ possesses a local dataset $\mathcal{D}_k$ with size $|\mathcal{D}_k|=d_k$, where each data sample is denoted by $\zeta \in \mathcal{D}_k$. For a given FL model $\bm{\omega}\in\mathbb{R}^n$, the loss function evaluated on $\zeta$ is denoted by $\ell(\bm{\omega};\zeta)$. FL is performed with the target of minimizing the  empirical loss over all local datasets, that is
\begin{align}
	\hat{l}(\bm{\omega}) = \frac{1}{d}\sum_{k=1}^K\sum_{\zeta\in\mathcal{D}_k}\ell(\bm{\omega},\zeta)=\sum_{k=1}^K\frac{d_k}{d}\hat{l}_k(\bm{\omega}),
\end{align} 
where $d=\sum_{k=1}^Kd_k$ is the total number of data, and $\hat{l}_k(\bm{\omega}) = \sum_{\zeta \in \mathcal{D}_k}\frac{1}{d_k}\ell(\bm{\omega},\zeta)$ denotes the local loss of client $k$.

Generally, the FL training process in one round consists of four steps, including local model update, local model upload, global model update, and global model download.

\textit{\textbf{Local Model Update}}:  At the beginning of each training round $t$, the clients receive the aggregated model $\bm{\omega}_{t}$ from the server. Each client implements $E$ epochs of the batch gradient descent (BGD) method. Specifically, the local model is initialized as $\bm{\omega}_{t,0}^k = \bm{\omega}_{t}$. In each epoch $\tau$, $\tau \in[1,E]$, client $k$ samples a mini-batch $\bm{\zeta}_{t,\tau}^k$ of size $\tilde{d}$ from $\mathcal{D}_k$, then the local model is updated as 
\begin{align}
	\bm{\omega}_{t,\tau+1}^k = \bm{\omega}_{t,\tau}^k - \frac{\eta_t}{\tilde{d}}\sum_{\zeta \in \bm{\zeta}_{t,\tau}^k}\nabla \ell(\bm{\omega}_{t,\tau}^k, \zeta),
\end{align}
where $\eta_t$ is the learning rate of FL.

\textit{\textbf{Local Model Upload}}:
After $E$ local update epochs, each client $k$ transmits its updated model $\bm{\omega}_{t,E}^{k}$ to the central server.

\textit{\textbf{Global Model Update}}:
The central server aggregates the model parameters $\bm{\omega}_{t,E}^k$ through weighted averaging as
\begin{align} \label{model_a}
	\bm{\omega}_{t+1}=\sum_{k=1}^{K}\frac{d_k}{d}\bm{\omega}_{t,E}^{k}.
\end{align}

\textit{\textbf{Global Model Download}}: The model $\bm{\omega}_{t+1}$ is broadcast to the clients for training round $t+1$.

In FL, personalized data collection across clients leads to heterogeneous local data distributions, which can degrade model performance.
Additionally, repeated local training and model transmission incur significant computation and communication overhead.
Motivated by these observations, we next model data heterogeneity and training cost to jointly address these two bottlenecks in effective FL.
\subsection{Data Distribution Heterogeneity Model}
To capture the statistics of data heterogeneity at the population level, we explicitly model local data distributions as following an underlying data generation process.
Let $\mathcal{Z}$ denote the data domain. Let $\mathcal{Q}$ denote the set of all local data distributions defined over $\mathcal{Z}$, and let $\mathcal{P}$ denote a meta-distribution over $\mathcal{Q}$. The global data distribution is defined as $p_g = \mathbb{E}_{p \sim \mathcal{P}}[p]$.
For each client $k \in \mathcal{K}$, its local distribution $p_k$ is drawn from $\mathcal{P}$, and the local dataset $\mathcal{D}_k$ consists of samples $\zeta \in \mathcal{Z}$ drawn from $p_k$. The dataset of all clients is denoted by $\mathcal{D}=\mathcal{D}_1\cup \mathcal{D}_2\cup\cdots\cup \mathcal{D}_K$.
By characterizing the dynamic nature of client distributions while remaining agnostic to any specific individual client, this modeling approach is more adaptable to client dynamics in large-scale wireless deployments \cite{hu2023generalization}.
Based on this model, the data distribution heterogeneity, also called data distribution divergence, is measured by the Kullback-Leibler (KL) divergence between the global distribution $p_g$ and the local distribution $p_k$, denoted by $D_{kl}(p_g \| p_k), \forall k \in \mathcal{K}$. Its detailed definition and impact on the global model performance will be explained in Section III.

\subsection{Training Cost Model}
\subsubsection{Computation Cost}
Given the limited local computation capability, the computation energy and latency for local model updates require consideration. Let $s_k$ and $f_{k,t}$ denote the number of CPU cycles required to process one bit and the CPU clock speed of device $k$, respectively. The computation latency for local model update is formulated as \cite{10057044}
\begin{align}
	{\mathcal{T}}_{k,t}^{com}=\frac{Es_kd_{k}}{f_{k,t}}.
\end{align}
Let $\epsilon$ represent the effective switched capacitance determined by the chip architecture. The computation energy per bit in one local epoch is $\epsilon s_k f_{k,t}^2$ \cite{10057044}. Therefore, the energy consumption of $E$ local model updates is accordingly calculated by
\begin{align}
	{\mathcal{E}}_{k,t}^{com}=d_k\epsilon s_k f_{k,t}^2E.
\end{align}

\subsubsection{Communication Cost} We now consider the costs incurred by frequent model exchange between clients and the server. Let $B$ denote the total uplink bandwidth. The proportion of uplink bandwidth allocated to client $k$ is denoted by $b_{k,t}$. 
Let $h_{k,t}$ be the small-scale fading, which follows a complex Gaussian distribution. 
The large-scale path loss is $\theta_k = (\frac{c_0}{4\pi f_c})^2({dist_k})^{-\gamma}$ \cite{5226961}, where $c_0$ is the speed of light, $f_c$ is the carrier frequency, $dist_k$ is the distance from client $k$ to the server, $\gamma$ is the path loss exponent. The transmit power is $P_{k}$, and the power spectrum density of
noise is $N_0$. 
Then, the uplink channel transmission rate of client $k$ in round $t$ is represented by
\begin{align}
	r_{k,t}=b_{k,t}Blog_2\left(1+\frac{\theta_k|h_{k,t}|^2P_{k}}{BN_0}\right).
\end{align}
Let $C_k$ denote the size of the local model parameters of client $k$. The transmission latency for client $k$ to upload its local model to the server is given by
\begin{align}
	{\mathcal{T}}_{k,t}^{up}=\frac{C_{k}}{r_{k,t}}.
\end{align}
The transmission energy of client $k$ is expressed as
\begin{align}
	{\mathcal{E}}_{k,t}^{up}=P_k{\mathcal{T}}_{k,t}^{up}.
\end{align}
The downlink latency of broadcasting the global model from the server to all clients is negligible compared with the uplink transmission latency of uploading local models \cite{vehicle_selection_fl}.

%

\section{Impact of data Heterogeneity on FL}\label{sec:Impact of data heterogeneity}
As explained in the Introduction, the effectiveness of FL in wireless networks relies on both training costs and global model generalization performance, which can be represented by global model generalization error. Accordingly, this section will theoretically analyze the impact of data distribution heterogeneity on generalization performance, which is critical for optimizing FL training cost. Moreover, the analysis will also consider the influence of the total dataset size $d$.
\subsection{Preliminaries}

\begin{definition} \label{definition dataset}
We index all data samples in $\mathcal{D}$ as
$\mathcal{D} = \{\zeta^{(1)}, \zeta^{(2)}, \ldots, \zeta^{(d)}\}$.
Let $\mathcal{D}^{(i)}$ denote a neighboring dataset obtained by replacing
the $i$-th data sample $\zeta^{(i)}$ with $\hat{\zeta}^{(i)}$, where
$\hat{\zeta}^{(i)} \in \mathcal{Z}$ and $\hat{\zeta}^{(i)} \neq \zeta^{(i)}$.
Let $\mathcal{D}'=\mathcal{D}'_1\cup \mathcal{D}'_2\cup\cdots\cup \mathcal{D}'_K$ be a dataset independently sampled from the same underlying
local data distributions $\{p_1, p_2, \ldots, p_K\}$ and sharing the same local
dataset sizes as $\mathcal{D}$.
Correspondingly, the models trained on datasets $\mathcal{D}^{(i)}$ and
$\mathcal{D}'$ at round $t$ are denoted by $\bm{\omega}_t^{(i)}$ and
$\bm{\omega}_t'$, respectively.
Moreover, the empirical losses evaluated on $\mathcal{D}^{(i)}$ and
$\mathcal{D}'$ are denoted by $\hat{l}_{\mathcal{D}^{(i)}}(\cdot)$ and
$\hat{l}_{\mathcal{D}'}(\cdot)$, respectively.
\end{definition}
\begin{definition}
	The population loss $l(\bm{\omega})$ measures the expected loss of the model $\bm{\omega}$ over the global data distribution $p_g$, and is defined as
 $l(\bm{\omega}):=\mathbb{E}_{\zeta \sim p_g}[\ell(\bm{\omega}, \zeta)]$.
\end{definition}
%
%
\begin{definition}\label{def_semi}
	The semi-empirical loss is defined as $\overline{l}(\bm{\omega}) := \frac{1}{d} \sum_{k=1}^K d_k l_k(\bm{\omega})$, where $l_k(\bm{\omega}) = \mathbb{E}_{\zeta\sim p_k} [\ell (\bm{\omega};\zeta)]$ \cite{hu2023generalization}. 
\end{definition}
\begin{remark}
	The semi-empirical loss $\overline{l}(\bm{\omega})$ lies between the empirical and population losses.
	It replaces the empirical loss with expectations over local data distributions of all clients, while retaining practical client participation through the aggregated weights $\{d_k/d\}_{k=1}^K$.
	Thus, $\overline{l}(\bm{\omega})$ serves as an intermediate quantity that captures client distribution heterogeneity without conditioning on a specific dataset.
\end{remark}

\begin{definition}\label{def 3}
	The generalization error is defined as $\tilde e(\bm{\omega}):= l(\bm{\omega}) - \hat{l}(\bm{\omega})$, which measures the gap between the population loss and the empirical loss.
\end{definition}
To enable the theoretical analysis, we make the following assumptions, which are commonly adopted in \cite{SunMimic2024,gu2021fast,cho2022towards, li2019convergence,bousquet2002stability}.
\begin{assumption}\label{asm:iid-clients}
	The local data distributions $\{p_k\}_{k\in\mathcal{K}}$ are mutually independent.
    Data samples in each local dataset are independently drawn from the corresponding local distribution.
\end{assumption}
\begin{assumption}\label{ass:smooth}
For all $\zeta\in\mathcal{Z}$, the loss function $\ell(\cdot,\zeta)$
	is differentiable, $\lambda$-smooth, and $L$-Lipschitz over the model space $\Theta \subseteq \mathbb{R}^n$, i.e.,$
	\ell(\bm{u};\zeta)
	\le \ell(\bm{v};\zeta)
	+ (\bm{u}-\bm{v})^{\mathsf T}\nabla \ell(\bm{v};\zeta)
	+ \frac{\lambda}{2}\|\bm{u}-\bm{v}\|^2$,
	and
	$
	|\ell(\bm{u};\zeta)-\ell(\bm{v};\zeta)|
	\le L\|\bm{u}-\bm{v}\|,
	\quad \forall \bm{u},\bm{v}\in\Theta,\ \forall \zeta\in\mathcal{Z}.
	$
\end{assumption}

\begin{assumption}\label{bounded_loss}
	For all $\zeta\in\mathcal{Z}$ and $\bm{\omega}\in\Theta$, the loss function $\ell(\bm{\omega};\zeta)$ is bounded, i.e., $ 0 \le \ell(\bm{\omega};\zeta) \le c $.
\end{assumption}
\begin{assumption}\label{ass:stable}
The FL procedure is $\beta$-uniformly stable, i.e., 
$
\big|\ell(\bm{\omega}, \zeta)
- \ell(\bm{\omega}^{(i)}, \zeta)\big|
\le \beta, \forall \zeta \in \mathcal{Z}
$.
\end{assumption}

\begin{assumption}\label{ass sgd}
	The stochastic gradient $\frac{1}{\tilde{d}}\sum_{\zeta \in \bm{\zeta}^k}\nabla \ell(\bm{\omega}, \zeta)$ on batch $\bm{\zeta}^k\in\mathcal{D}_k$ is an unbiased estimation of full batch, and its variance is bounded, i.e., $\mathbb{E}[\frac{1}{\tilde{d}}\sum_{\zeta \in \bm{\zeta}^k}\nabla \ell(\bm{\omega}, {\zeta})]=\nabla \hat{l}_k(\bm{\omega})$, and $\mathbb{E}||\frac{1}{\tilde{d}}\sum_{\zeta \in \bm{\zeta}^k}\nabla \ell(\bm{\omega}, \zeta)-\nabla \hat{l}_k(\bm{\omega})||^2\leq\sigma^2, \forall k \in \mathcal{K}$.
\end{assumption}

\subsection{Theoretical Analysis}
\begin{lemma}\label{lemma1}	
	With Assumption \ref{bounded_loss}, define  \(\tilde e_k := \hat{l}_k'(\cdot)-l_k(\cdot)\), where $\hat{l}'_k(\cdot) = \sum_{\zeta \in \mathcal{D}'_k}\frac{1}{d_k}\ell(\cdot,\zeta)$
    , then the following inequality holds for all \(k \in \mathcal{K}\):
	\begin{align} \label{DV3}
		\log\mathbb{E}_{\zeta\sim p_k}\left[\exp({\frac{d_k\tilde e_k}{d}})\right] \geq \frac{d_k}{d} \mathbb{E}_{\zeta\sim p_g}[\tilde e_k] - D_{kl}(p_g \| p_k).
	\end{align}
\end{lemma}
\begin{proof}
	Let $q_k={d_k\tilde{e}_k}/{d}$. Under Assumption \ref{bounded_loss}, $q_k$ is bounded as $|q_k|\leq {d_kc}/{d}$, which satisfies the requirement of Donsker and Varadhan's variational formula \cite{donsker1983asymptotic}. 
	Therefore, for any $k \in \mathcal{K}$, the KL divergence between $p_k$ and $p_g$ is given by
	\begin{align}\label{dv}
		D_{kl}(p_k \| p_k) &= \sup_{q_k} (\mathbb{E}_{\zeta \sim p_g}[q_k] - \log\mathbb{E}_{\zeta \sim p_k}[\exp(q_k)]) \nonumber\\
		&\geq \mathbb{E}_{\zeta \sim p_g}[q_k] - \log\mathbb{E}_{\zeta \sim p_k}[\exp(q_k)].
	\end{align}
	Substituting $q_k={d_k\tilde{e}_k}/{d}$ into (\ref{dv}) completes the proof.
\end{proof}


\begin{lemma}\label{SGD}
	With Assumptions \ref{ass:smooth} and \ref{ass sgd}, the difference between FL global models in rounds $t$ and $t-1$, 
    satisfies
	\begin{align}\label{one_step}
		\mathbb{E}||\bm{\omega}_{t}-\bm{\omega}_{t-1}||\leq2(E-1)\left(\frac{2\eta_{t-1}^2\lambda}{E}\Delta_{t-1}+\frac{\eta^2_{t-1}\sigma^2}{E^2}\right),
	\end{align}
	where $\Delta_{t-1}=\mathbb{E}[\hat{l}(\bm{\omega}_{t-1})-\sum_{k=1}^K\frac{d_k}{d}\hat{l}_k(\bm{\omega}^*_k)]$.
	\begin{proof}
		From~\eqref{model_a}, $	\bm{\omega}_{t}=\sum_{k=1}^{K}\frac{d_k}{d}\bm{\omega}_{{t-1},E}^{k}$, we have
		\begin{align}
			&\mathbb{E}||\bm{\omega}_{t}-\bm{\omega}_{t-1}||
			=\mathbb{E}||\sum_{k=1}^K\frac{d_k}{d}(\bm{\omega}_{{t-1},E}^k-\bm{\omega}_{t-1})||\nonumber\\
			\overset{(a)}{\leq}&2(E-1)\left(\frac{2\eta_{t-1}^2}{E}\sum_{k=1}^K\frac{d_k}{d}\mathbb{E}||\nabla\hat{l}_k(\bm{\omega}_{t-1})||^2+\frac{\eta^2_{t-1}\sigma^2}{E^2}\right)\nonumber\\
			\overset{(b)}\leq&2(E-1)\left(\frac{2\eta_{t-1}^2\lambda}{E}\sum_{k=1}^K\frac{d_k}{d}\mathbb{E}(\hat{l}_k(\bm{\omega}_{t-1})-\hat{l}_k(\bm{\omega}^*_k))+\frac{\eta^2_{t-1}\sigma^2}{E^2}\right)\nonumber\\
			=&2(E-1)\left(\frac{2\eta_{t-1}^2\lambda}{E}\Delta_{t-1}+\frac{\eta^2_{t-1}\sigma^2}{E^2}\right),
		\end{align}
		where (a) follows from Assumption \ref{ass sgd} and Lemma C.2 in \cite{gu2021fast}, and (b) follows from \(\lambda\)-smooth in Assumption \ref{ass:smooth}.
	\end{proof}
\end{lemma}

\begin{lemma}\label{lemma_stability}Under Assumption 4, for any FL round $t$, the upper bound of $\mathbb{E}[\hat{l}_{\mathcal{D}'}(\bm{\omega}_t)-\hat{l}(\bm{\omega}_t)]$ satisfies
	\begin{align}
		&\mathbb{E}[\hat{l}_{\mathcal{D}'}(\bm{\omega}_t)-\hat{l}(\bm{\omega}_t)]\nonumber\\
		\leq&\sum_{t'=0}^{t-1}4(E-1)\left(\frac{2\eta_{t'}^2\lambda L}{E}\Delta_{t'}+\frac{\eta^2_{t'}\sigma^2L}{E^2}\right).			
	\end{align}
	\begin{proof}
		We start with the decomposition
		\begin{align}\label{decome}
			&\mathbb{E}[\hat{l}_{\mathcal{D}'}(\bm{\omega}_t)-\hat{l}(\bm{\omega}_t)]\nonumber\\
			=&\mathbb{E}[\hat{l}_{\mathcal{D}'}(\bm{\omega}_t)-\hat{l}_{\mathcal{D}'}(\bm{\omega}'_t)+\hat{l}_{\mathcal{D}'}(\bm{\omega}'_t)-\hat{l}(\bm{\omega}_t)]\nonumber\\
			=&\underbrace{\mathbb{E}[\hat{l}_{\mathcal{D}'}(\bm{\omega}_t)-\hat{l}_{\mathcal{D}'}(\bm{\omega}'_t)]}_{A1}+\underbrace{\mathbb{E}[\hat{l}_{\mathcal{D}'}(\bm{\omega}'_t)-\hat{l}(\bm{\omega}_t)]}_{A2}.
		\end{align}
		
		As in Definition \ref{definition dataset}, since $\mathcal{D}$ and $\mathcal{D}'$ are independently sampled from the same distribution, 
        $A2$ in (\ref{decome}) is derived as
		\begin{align}\label{zero}
			\mathbb{E}[\hat{l}_{D'}(\bm{\omega}_t')-\hat{l}(\bm{\omega}_t)]=0.
		\end{align}	
		
		We now bound $A1$ in \eqref{decome}. 		Assumption~\ref{ass:smooth} ensures that the loss is
		$L$-Lipschitz, thus
		\begin{align} \label{eq_lip}
			\mathbb{E}|\hat{l}_{\mathcal{D}'}(\bm{\omega}_t)
			-\hat{l}_{\mathcal{D}'}(\bm{\omega}'_t)|
			\le 
			L \mathbb{E}\|\bm{\omega}_t-\bm{\omega}'_t\|.
		\end{align} 
        We first adding and subtracting $\bm{\omega}_{t-1}$ in $\mathbb{E}\|\bm{\omega}_t-\bm{\omega}'_t\|$, and decompose it by the triangle inequality as
		\begin{align}\label{tri_rec}
			&\mathbb{E}\|\bm{\omega}_t-\bm{\omega}'_t\|\nonumber\\
            =&\mathbb{E}\|\bm{\omega}_t-\bm{\omega}_{t-1}
			+
			\bm{\omega}_{t-1}-\bm{\omega}'_{t-1}
			+
			\bm{\omega}'_{t-1}-\bm{\omega}'_{t}\| \nonumber\\
			\le&
			\mathbb{E}\|\bm{\omega}_t-\bm{\omega}_{t-1}\|
			+
			\mathbb{E}\|\bm{\omega}'_t-\bm{\omega}'_{t-1}\|
			+
			\mathbb{E}\|\bm{\omega}_{t-1}-\bm{\omega}'_{t-1}\|.
		\end{align}
		
		By Lemma~\ref{SGD}, 
		the bound holds for both $\mathbb{E}\|\bm{\omega}_t-\bm{\omega}_{t-1}\|$ and
		$\mathbb{E}\|\bm{\omega}'_t-\bm{\omega}'_{t-1}\|$. 
        Substituting \eqref{one_step} into \eqref{tri_rec}, we obtain
\begin{align}
\mathbb{E}\|\bm{\omega}_t-\bm{\omega}'_t\|
&\le
\mathbb{E}\|\bm{\omega}_{t-1}-\bm{\omega}'_{t-1}\|
\nonumber\\
&+4(E-1)\!\left(
\frac{2\eta_{t-1}^2\lambda}{E}\Delta_{t-1}
+
\frac{\eta_{t-1}^2\sigma^2}{E^2}\right).
\label{recursion}
\end{align}
Iterating \eqref{recursion} from $0$ to $t-1$ gives
\begin{align}\label{model_diff}
\mathbb{E}\|\bm{\omega}_t-\bm{\omega}'_t\|
&\le
\mathbb{E}\|\bm{\omega}_{0}-\bm{\omega}'_{0}\|\nonumber\\
&+
\sum_{t'=0}^{t-1}
4(E-1)\!\left(
\frac{2\eta_{t'}^2\lambda}{E}\Delta_{t'}
+
\frac{\eta_{t'}^2\sigma^2}{E^2}\right).
\end{align}
Note that the fact when $\bm{\omega}_{0}=\bm{\omega}_{0}'$, (\ref{model_diff}) yields
\begin{align}\label{model_diff1}
        \mathbb{E}\|\bm{\omega}_t-\bm{\omega}'_t\|
        \le
        \sum_{t'=0}^{t-1}
        4(E-1)\!\left(
        \frac{2\eta_{t'}^2\lambda}{E}\Delta_{t'}
        +
        \frac{\eta_{t'}^2\sigma^2}{E^2}
        \right)
\end{align}
		Substituting \eqref{model_diff1} into \eqref{eq_lip}, $A1$ yields
		\begin{align}\label{a1}
			&\mathbb{E}\big|\hat{l}_{\mathcal{D}'}(\bm{\omega}_t)
			-\hat{l}_{\mathcal{D}'}(\bm{\omega}'_t)\big|\nonumber\\
			\le&
			\sum_{t'=0}^{t-1}
			4(E-1)\!\left(
			\frac{2\eta_{t'}^2\lambda L}{E}\Delta_{t'}
			+
			\frac{\eta_{t'}^2\sigma^2L}{E^2}
			\right).
		\end{align}
		Combining with \eqref{zero} and \eqref{a1} completes the proof.
	\end{proof}
	
\end{lemma}

\begin{theorem}\label{generalization}
	Under Assumptions \ref{asm:iid-clients}--\ref{ass sgd}, 
    for any $\delta \in (0,1)$, with probability at 
	least $1-\delta$ over the random draw of dataset $\mathcal{D}$, {\bf the 
	generalization error $\tilde e(\bm{\omega}_t)$ of the global model} at training round $t$
    , satisfies
	\begin{align} \label{gen}
		&\tilde e(\bm{\omega}_t)
		\leq \sum_{t'=0}^{t-1}4(E-1)\left(\frac{2\eta_{t'}^2\lambda L}{E}\Delta_{t'}+\frac{\eta^2_{t'}\sigma^2L}{E^2}\right)\nonumber\\
		&+\sqrt{\frac{c^2 \log \frac{4}{\delta}}{2}}\sigma_d^2 + \sum_{k=1}^{K} D_{kl}(p_g \| p_k) + {\frac{c^2}{8d}}+C_{\beta}.
	\end{align}
	where $\sigma_d^2 = {\sum_{k=1}^{K} \sqrt{d_k}}/{d}$, and $C_\beta = (2\beta+{c}/{d})\sqrt{d\ln({2}/{\delta})}$.
\end{theorem}


\begin{proof}
	Recall Definition \ref{def 3}, where $\tilde{e}(\bm{\omega}_t)=l(\bm{\omega}_t) - \hat{l}(\bm{\omega}_t)$. 
    By introducing the intermediate term $\hat{l}_{\mathcal{D}'}(\bm{\omega}_t)$,
    $\tilde{e}(\bm{\omega}_t)$ can be decomposed as
	\begin{align}\label{e}
		\tilde{e}(\bm{\omega}_t)=\underbrace{l(\bm{\omega}_t) - \hat{l}_{\mathcal{D'}}(\bm{\omega}_t)}_{\tilde{e}'(\bm{\omega}_t)}+\underbrace{\hat{l}_{\mathcal{D'}}(\bm{\omega}_t)-\hat{l}(\bm{\omega}_t)}_{\tilde{e}''(\bm{\omega}_t)},	
	\end{align}
	where $\tilde e'(\bm{\omega}_t)$ denotes the generalization-related error and $\tilde{e}''(\bm{\omega}_t)$ represents the training-related error. Given a specific training dataset, in the following, we show that $\tilde{e}'$ is caused by the inaccessibility of the broader dataset and reflects limited training samples and client distribution heterogeneity, whereas $\tilde{e}''$ is induced by the intrinsic factors of FL training process, such as the learning rate, algorithmic stability and the number of local epochs.
	
	\subsubsection{Analyze  $\tilde{e}'(\bm{\omega}_t)$} Due to training independence of $\tilde e'(\bm{\omega}_t)$, we rewrite it as $\tilde e' = l(\cdot) - \hat{l}_{\mathcal{D}'}(\cdot)$. By introducing the semi-empirical loss $\overline{l}(\cdot)$ in Definition \ref{def_semi}, $\tilde e'$ can be divided into
	\begin{align}
		\tilde e' = \underbrace{l(\cdot) - \overline{l}(\cdot)}_{\text{A3}} + \underbrace{\overline{l}(\cdot) - \hat{l}_{\mathcal{D}'}(\cdot)}_{\text{A4}},
	\end{align}
	where $A3$ denotes the distribution error and $A4$ the sample error. The distribution error measures the gap between the population and semi-empirical losses induced by the divergence between local and global distribution across clients. The sample error captures the difference between the semi-empirical and empirical losses due to the finite local data samples.
	
	Let us now analyze $A3$. By the definition $\tilde{e}_k=\hat{l}_k'(\cdot)-l_k(\cdot)$ in Lemma \ref{lemma1},
	\begin{align}
		&\mathbb{E}_{\zeta\sim p_1,\dots,\zeta\sim p_K}\!\left[
		\exp\!\left(\sum_{k=1}^K \frac{d_k}{d}\,\tilde e_k \right)
		\right]
		\overset{(a)}{=}
		\prod_{k=1}^{K}
		\mathbb{E}_{\zeta \sim p_k}\!\left[
		\exp\!\left(\frac{d_k}{d}\,\tilde e_k\right)
		\right] \nonumber\\[6pt]
		&\overset{(b)}{=}
		\prod_{k=1}^{K}
		\mathbb{E}_{\zeta \sim p_k}\!\left[
		\exp\!\left(\frac{1}{d}
		\sum_{\zeta \in \mathcal{D}'_k} \big(\ell(\cdot;\zeta)- l_k(\cdot)\big)
		\right)
		\right] \nonumber\\[6pt]
		&\overset{(c)}{=}
		\prod_{k=1}^{K}
		\prod_{\zeta \in \mathcal{D}'_k}
		\mathbb{E}_{\zeta \sim p_k}
		\!\left[
		\exp\!\left(
		\frac{1}{d}\big(\ell(\cdot;\zeta)-l_k(\cdot)\big)
		\right)
		\right] \nonumber\\[6pt]
		&\overset{(d)}{\le}
		\prod_{k=1}^{K}
		\left(
		\exp\!\left(\frac{c^2}{8 d^2}\right)
		\right)^{d_k}
		= 
		\exp\!\left(
		\frac{c^2}{8 d^2}\sum_{k=1}^K d_k
		\right)
		=
		\exp\!\left(\frac{c^2}{8 d}\right).
		\label{8n}
	\end{align}
	where $(a)$ and $(c)$ follow from Assumption \ref{asm:iid-clients}, $(b)$ follows from $\hat{l}'_k(\cdot) = \sum_{\zeta \in \mathcal{D}'_k}\frac{1}{d_k}\ell(\cdot,\zeta)$, while $(d)$ follows from Hoeffding's Lemma  \cite{hoeffding1963probability}, and the fact that $\mathbb{E}[\hat{l}'_k(\cdot)]=l_k(\cdot)$.
	%
	
	By summing up both sides over all $k \in \mathcal{K}$ in (\ref{DV3}), we obtain
	\begin{align}\label{lemma1.1}
		&\sum_{k=1}^K\frac{d_k}{d}\mathbb{E}_{\zeta\sim p_g}[\tilde e_k] - \sum_{k=1}^KD_{kl}(p_g \| p_k)\nonumber\\
		\leq&\log\mathbb{E}_{\zeta\sim p_1, \dots, \zeta\sim p_K}\left[\exp\left(\sum_{k=1}^K{\frac{d_k\tilde e_k}{d}}\right)\right].   
	\end{align}
    Noting that for all $k \in \mathcal{K}$, when $p_k$ coincides with $p_g$, the expected empirical loss satisfies $\mathbb{E}_{\zeta\sim {p_g}}[\hat{l}'_k(\cdot)] = l(\cdot)$, which implies that $\mathbb{E}_{\zeta\sim p_g}[\tilde e_k]=l(\cdot)-\overline{l}_k(\cdot)$. Since $\sum_{k=1}^K\frac{d_k}{d}\overline{l}_k(\cdot) = \overline{l}(\cdot)$, (\ref{lemma1.1}) then yields
	\begin{align}\label{logeee}
		&l(\cdot) - \overline{l}(\cdot) - 
		\sum_{k=1}^{K} D_{kl}(p_g \| p_k)\nonumber\\
		&\leq\log\mathbb{E}_{\zeta\sim p_1, \dots, \zeta\sim p_K}\left[\exp\left({\sum_{k=1}^{K}\frac{d_k\tilde e_k}{d}}\right)\right] .
	\end{align}
	Substituting (\ref{8n}) into (\ref{logeee}) leads to
	\begin{align} \label{di}
		l(\cdot) - \overline{l}(\cdot) \leq {\frac{c^2}{8d}} + \sum_{k=1}^{K} D_{kl}(p_g \| p_k).
	\end{align}
	%
	
	Let us analyze $A4$.
	By applying Hoeffding's Lemma to each $k \in \mathcal{K}$, for any $\delta\in(0,1)$,
	with probability taken over the randomness of sampling
	in $\mathcal{D}'$, at least \(1 - \delta\), $A4$ is bounded as
	\begin{align} \label{em}
		\overline{l}(\cdot) - \hat{l}_{\mathcal{D}'}(\cdot)
		&= \sum_{k=1}^K \frac{d_k}{d}\bigl(l_k(\cdot)-\hat{l}'_k(\cdot)\bigr)\nonumber\\
		&\leq \sum_{k=1}^K \frac{d_k}{d}
		\sqrt{\frac{c^2}{2d_k}\log\frac{2}{\delta}}\nonumber\\
		&= \sqrt{\frac{c^2}{2}\log \frac{2}{\delta}}\,
		\sum_{k=1}^K\frac{\sqrt{d_k}}{d}.
	\end{align}

	Combining (\ref{di}) and (\ref{em}), the upper bound of $\tilde e'$ becomes
	\begin{align}\label{e1}
		\tilde e' \leq \sigma_d^2 \sqrt{\frac{c^2}{2} \log \frac{2}{\delta}} + \sum_{k=1}^{K} D_{kl}(p_g \| p_k) + {\frac{c^2}{8d}}.
	\end{align}
	which holds with probability at least $1-\delta$.

	\subsubsection{Analyze $\tilde{e}''(\bm{\omega}_t)$} 
	In round $t$, we first consider the case where the replacement of $\zeta^{(i)} \in \mathcal{D}$ induces a perturbation in model $\bm{\omega}_t$, which in turn leads to a change in the empirical loss difference $\hat{l}_{\mathcal{D}'}(\bm{\omega}_t)-\hat{l}(\bm{\omega}_t)$. We define $\Delta \ell ^{(i,j)}_t = \ell(\bm{\omega}_t^{(i)}; \zeta^{(j)})-\ell(\bm{\omega}_t; \zeta^{(j)})$ and $\Delta \ell^{(i,i)}_t = \ell(\bm{\omega}_t^{(i)}; \zeta^{(i)})-\ell(\bm{\omega}_t; \hat{\zeta}^{(i)})$. Under Assumptions \ref{bounded_loss} and \ref{ass:stable}, 
    we have $\Delta \ell^{(i,j)}_t \leq \beta$ and $\Delta \ell^{(i,i)}_t \leq c$. Hence, it follows that,
	\begin{align}
		&(\hat{l}_{\mathcal{D}'}(\bm{\omega}_t)-\hat{l}(\bm{\omega}_t))-(\hat{l}_{\mathcal{D}'}(\bm{\omega}_t^{(i)})-\hat{l}_{\mathcal{D}^{(i)}}(\bm{\omega}_t^{(i)}))\nonumber\\
		&\leq|\hat{l}_{\mathcal{D}'}(\bm{\omega}_t)-\hat{l}_{\mathcal{D}'}(\bm{\omega}_t^{(i)})|+|\hat{l}_{\mathcal{D}^{(i)}}(\bm{\omega}_t^{(i)})-\hat{l}(\bm{\omega}_t)|\nonumber\\
		&\leq\beta+\frac{1}{d}|\sum_{j=1, j\neq i}^d\Delta \ell ^{(i,j)}_t+\Delta \ell^{(i,i)}_t|\nonumber\\
		&\leq\beta+\frac{1}{d}|\sum_{j=1, j\neq i}^d\Delta \ell ^{(i,j)}_t|+\frac{1}{d}|\Delta \ell^{(i,i)}_t|\nonumber\\
		&\leq \beta+\frac{d-1}{d}\beta+\frac{c}{d}\leq 2\beta+\frac{c}{d}				
	\end{align}
	We next consider the replacement of $\zeta^{'(i)} \in \mathcal{D}'$ by $\hat{\zeta}^{(i)'} \in \mathcal{Z}$, which induces no perturbation in $\bm{\omega}_t$, and denote the resulting dataset by $\mathcal{D}^{'(i)}$. 
	Define $\Delta \ell^{(i)}_t=\ell(\bm{\omega}_t;\zeta^{'(i)})-\ell(\bm{\omega}_t;\hat{\zeta}^{'(i)})$. Assumption \ref{bounded_loss} implies $\Delta \ell^{(i)}_t \leq c $, and hence
	\begin{align}
		&\hat{l}_{\mathcal{D}'}(\bm{\omega}_t)-\hat{l}(\bm{\omega}_t)
		-(\hat{l}_{\mathcal{D}'^{(i)}}(\bm{\omega}_t)-\hat{l}(\bm{\omega}_t))\nonumber\\
		&=\hat{l}_{\mathcal{D}'}(\bm{\omega}_t)-\hat{l}_{\mathcal{D}'^{(i)}}(\bm{\omega}_t)\nonumber\\
		&= \frac{1}{d}\sum_{j=1, j\neq i}^d
		\big(\ell(\bm{\omega}_t;\zeta^{'(j)})-\ell(\bm{\omega}_t;\zeta'^{(j)})\big)+\frac{1}{d}\Delta \ell^{(i)}_t\nonumber\\
		&\le \frac{c}{d}.
	\end{align}
	%


	Note that for all $\zeta \in \mathcal{D}\cup\mathcal{D'}$, there exists a constant $2\beta+c/d = \max(2\beta+c/d, c/d)$ such that McDiarmid's inequality holds \cite{bousquet2002stability}. Thus, for any $\delta\in(0,1)$ over the random draws of the $\mathcal{D}$ and $\mathcal{D}'$, with probability at least $1-\delta$, we have
	\begin{align}\label{e2}
		&\hat{l}_{\mathcal{D}'}(\bm{\omega}_t)-\hat{l}(\bm{\omega}_t)\leq(2\beta+\frac{c}{d})\sqrt{d\ln\frac{1}{\delta}}\nonumber\\
		&+\sum_{t'=0}^{t-1}4(E-1)\left(\frac{2\eta_{t'}^2\lambda L}{E}\Delta_{t'}+\frac{\eta^2_{t'}\sigma^2L}{E^2}\right).	
	\end{align}
	To substitute (\ref{e1}) and (\ref{e2}) into (\ref{e}), we conclude the proof.
	
\end{proof}
\begin{remark}
	The first term on the RHS of (\ref{gen}) captures the accumulated loss estimation drift introduced by performing the BGD algorithm. $\sigma_d^2$ in the second term arises from the variation caused by imbalanced sample sizes across clients. Note that $\sigma_d^2$ is minimized when one client holds all datasets ($d_k=d, d_{k'}=0, \forall k'\neq k$), which corresponds to the centralized training. $D_{kl}(p_g \| p_k)$ in the third term characterizes the degree of data heterogeneity across each client by measuring the divergence between the global distribution and each local distribution. The fourth term $c^2/8d$ reflects that large total sample sizes can effectively reduce generalization error. The last term $C_\beta$ indicates that algorithmic stability produces better generalization. Specifically, when the FL algorithm exhibits high uniform stability, reducing the distribution divergence $D_{kl}(p_g \| p_k)$ for each client and enlarging the total dataset size $d$ both contribute to tightening the generalization error bound.
\end{remark}

\section {Problem Formulation and Solutions} \label{problem}
The analysis in Theorem \ref{generalization} indicates that the distribution divergence fundamentally influences the generalization performance of FL. 
Meanwhile, the frequent model exchanges and local model update 
introduces significant learning latency and energy consumption of FL. 
To jointly address these issues, a distribution divergence aware client selection and resource allocation method is designed to ensure model accuracy while optimizing the training costs. 

\subsection{Problem Formulation}

In each round $t$, let \(a_{k,t} = 1\) denote that client \(k\) is selected to participate FL; \(a_{k,t} = 0\), otherwise.  
After the local model from all selected clients are collected, the server starts to update the global model. The learning latency of FL execution in round $t$ is given by
\begin{align}
	{\mathcal{T}_t} = \max_{k}a_{k,t}{({\mathcal{T}}_{k,t}^{up} + {\mathcal{T}}_{k,t}^{com})}, \label{12}
\end{align}
and the energy consumption is
\begin{align}
	{\mathcal{E}_t}=\sum_{k=1}^{K} a_{k,t} ({\mathcal{E}}_{k,t}^{up} + {\mathcal{E}}_{k,t}^{com}).
\end{align}

To jointly optimize learning latency and energy consumption during FL, we introduce a multi-objective optimization utility function as in \cite{9681911}
\begin{align}
	Y_t = \alpha_1 {\mathcal{T}_t} + \alpha_2 {\mathcal{E}_t},
\end{align}
where \(\alpha_1, \alpha_2\) control the Pareto-optimal trade-offs between the learning latency and energy consumption.
Let client selection, bandwidth allocation, and computation capacity allocation decisions be represented as $\mathcal{A}_t = \{a_{1,t}, a_{2,t},..., a_{K,t}\}, \mathcal{B}_t = \{b_{1,t}, b_{2,t},...,b_{K,t}\}, \mathcal{F}_t = \{f_{1,t}, f_{2,t},...,f_{K,t}\}$, respectively. The optimization problem can be formulated as

\begin{subequations}
	\vspace{-10pt}
	\begin{align}
		\textbf{P1}~\min_{\mathcal{A}_t,\mathcal{B}_t,\mathcal{F}_t} Y_t,
	\end{align} 
	\begin{align}\label{c11}
		s.t. ~a_{k,t} \in \{0,1\}, \forall k \in \mathcal{K},
	\end{align}
	\begin{align}\label{c21}
		\sum_{k=1}^{K}b_{k,t}\leq 1,
	\end{align}
	\begin{align}\label{c31}
		0\leq f_{k,t}\leq f_k^{max}, \forall k \in \mathcal{K},
	\end{align}
	\begin{align}\label{c41}
		a_{k,t}D_{kl}(p_g \| p_k) \leq e_1^{max}, \forall k \in \mathcal{K},
	\end{align}
	\begin{align}\label{c51}
		{\sum_{k=1}^{K}a_{k,t}d_k} \geq e_2^{max}.
	\end{align}
\end{subequations}
(\ref{c11}) indicates whether client $k$ participates in round $t$. (\ref{c21}) denotes that the total allocated bandwidth cannot exceed the system limit. (\ref{c31}) enforces client $k$ performs local model update within its computation capability $f_k^{max}$. Building upon Theorem \ref{generalization}, constraints (\ref{c41}) and (\ref{c51}) are introduced to ensure the model performance. Constraint (\ref{c41}) restricts the distribution divergence of the selected clients to be below the KL threshold $e_1^{max}$, while (\ref{c51}) ensures that the selected clients possess sufficient dataset size beyond the data size budget $e_2^{max}$.

\subsection{Problem Solutions}
To solve $\textbf{P1}$, the global data distribution $p_g$ and the local distributions
$\{p_1,\ldots,p_K\}$ in constraint ($\ref{c41}$) need to be obtained.
Thus, each client first estimates its local distribution based on data category proportions. Suppose there are $Z$ categories. For client $k$, let $d_{z,k}$ denote the number of samples belonging to category $z$. The proportion of category $z$ is given by
\begin{align}
	\hat{p}_{z,k} = \frac{d_{z,k}}{d_k}.
\end{align}
Then, the estimated data distribution vector of client $k$ is
\begin{align}
	\hat{p}_k=[\hat{p}_{1,k}, \hat{p}_{2,k},...,\hat{p}_{Z,k}].
\end{align} 
Each client uploads its estimated local distribution $\hat{p}_k$ to the server. Recalling that $p_g = \mathbb{E}_{p \sim \mathcal{P}}[p]$, the server estimates the global distribution as
\begin{align}
	\hat{p_g}=\sum_{k=1}^K\frac{d_k}{d}\hat{p}_k.
\end{align} 
Based on constraint (\ref{c41}), the server selects a subset of clients whose distribution satisfy $D_{kl}(\hat{p_g} \| \hat{p}_k) \leq e_1^{max}, \forall k \in \mathcal{K}$. The resulting client set is denoted by $\mathcal{K}'$, and we set $a_{k,t}=0, \forall k \in \mathcal{K}\setminus\mathcal{K}'$. 
Thus, problem $\textbf{P1}$ can be transformed into
	
\begin{subequations}
	\begin{align}	  \textbf{P1}'~\min_{\mathcal{A}_t,\mathcal{B}_t,\mathcal{F}_t} Y_t,
	\end{align}
    \begin{align}
		s.t.(\ref{c11}),(\ref{c21}),(\ref{c31}),(\ref{c51}).
	\end{align} 	
\end{subequations}
\noindent $\textbf{P1}'$ is an MINLP problem \cite{9681911}. Due to its complexity, we introduce an auxiliary variable $\Upsilon_t$ and reformulate $\textbf{P1}$ in epigraph form \cite{vehicle_selection_fl}. Specifically, $\Upsilon_t$ represents an upper bound on $\alpha_1a_{k,t}{({\mathcal{T}}_{k,t}^{up} + {\mathcal{T}}_{k,t}^{com})}$ as expressed in (\ref{zeta2}). Accordingly, $\textbf{P1}'$ can be reformulated as
\begin{subequations}
	\begin{align}
		\textbf{P2}~\min_{\mathcal{A}_t,\mathcal{B}_t,\mathcal{F}_t,\Upsilon_t} &\Upsilon_t+\alpha_2 \sum_{k=1}^{K} a_{k,t} (\frac{P_kC_{k}}{b_{k,t}R_{k,t}} + d_k\epsilon s_k f_{k,t}^{2}E),\nonumber
	\end{align} 
	\begin{align}
		s.t.~(\ref{c11}),(\ref{c21}),(\ref{c31}),(\ref{c51}),\nonumber
	\end{align}
	\begin{align}\label{zeta2}
		\alpha_1 {\frac{a_{k,t}C_{k}}{b_{k,t}R_{k,t}}}+
		\alpha_1{ \frac{Ea_{k,t}s_kd_{k}}{f_{k,t}}}\leq\Upsilon_t, \forall k \in \mathcal{K}',
	\end{align}
\end{subequations}
where $R_{k,t}=Blog_2(1+{\theta_{k,t}|h_{k,t}|^2P_{k}}/{BN_0})$.

To solve $\textbf{P2}$, we develop a two-step joint client selection and resource allocation algorithm. Specifically, (a) we fix $\mathcal{F}_t$ and update $\mathcal{A}_t$ and $\mathcal{B}_t$, where the corresponding sub-problem can be reformulated as a difference-of-convex (DC) problem; (b) we update $\mathcal{F}_t$ while keeping $\mathcal{A}_t$ and $\mathcal{B}_t$ fixed, resulting in a convex problem.

Given a fixed $\mathcal{F}_t$, the sub-problem in step (a) is formulated as
\begin{subequations}
	\begin{align}
		\textbf{P3}~\min_{\mathcal{A}_t,\mathcal{B}_t,\Upsilon_t} &\Upsilon_t+\alpha_2 \sum_{k=1}^{K} a_{k,t} (\frac{P_kC_{k}}{b_{k,t}R_{k,t}} + d_k\epsilon s_k f_{k,t}^2E),\nonumber
	\end{align} 
	\begin{align}
		s.t.~(\ref{c11}),(\ref{c21}),(\ref{c51}),\nonumber
	\end{align}
	\begin{align}\label{zeta3}
		\alpha_1 {\frac{a_{k,t}C_{k}}{b_{k,t}R_{k,t}}}+
		\alpha_1{ \frac{Ea_{k,t}s_kd_{k}}{f_{k,t}}}\leq\Upsilon_t, \forall k \in \mathcal{K}'.
	\end{align}
\end{subequations}
The  product term ${a_{k,t}}/{b_{k,t}}$ renders $\textbf{P3}$ non-convex. To reformulate $\textbf{P3}$ into a more tractable form, we introduce auxiliary variables $z_{k,t}$, $u_{k,t}$.
Specifically, let
\begin{align}
	z_{k,t}=\frac{1}{b_{k,t}}, \qquad \forall k \in \mathcal{K}',
\end{align}
and define
\begin{align}\label{uk}
	u_{k,t}=a_{k,t} z_{k,t}, \qquad \forall k \in \mathcal{K}'.
\end{align}
Then, constraints (\ref{c21}) and (\ref{zeta3}) can be rewritten as
\begin{align}\label{z}
	\sum_{k=1}^{K} \frac{1}{z_{k,t}} \leq 1,
\end{align}
and
\begin{align}\label{zetaz}
	\alpha_1 \frac{C_{k} u_{k,t}}{R_{k,t}}
	+
	\alpha_1 \frac{a_{k,t} E s_k d_{k}}{f_{k,t}}
	\leq \Upsilon_t,
	\qquad \forall k \in \mathcal{K}'.
\end{align}
Further, by utilizing the reformulation-linearization technique (RLT) \cite{8607120}, (\ref{uk}) is equivalent to
the following constraints.
\begin{equation}\label{u}
	u_{k,t} =
	\begin{cases}
		u_{k,t} \geq a_{k,t}, & k \in \mathcal{K}'; \\
		u_{k,t} \leq \frac{a_{k,t}}{b^{min}}, & k \in \mathcal{K}'; \\
		u_{k,t} \leq z_{k,t}+a_{k,t}-1, & k \in \mathcal{K}';\\
		u_{k,t} \geq z_{k,t}-\frac{(1-a_{k,t})}{b^{min}}, & k \in \mathcal{K}',
	\end{cases}
\end{equation}
where $b^{min}$ is an infinitely small constant to ensure that the variable \( z_{k,t} \) is bounded, enabling effective linear relaxations can be constructed for \( u_{k,t} \).
%
When $a_{k,t} = 0$, the first two inequalities imply $0\leq u_{k,t} \leq 0$, which gives $u_{k,t}=0$. Since $1 \leq z_{k,t}\leq1/b^{min}$, the last two inequalities are trivially satisfied.  When $a_{k,t} = 1$, the last two inequalities yield $z_{k,t} \leq u_{k,t} \leq z_{k,t}$, which indicates $u_{k,t}=z_{k,t}$, while the first two inequalities hold automatically. Therefore, $u_{k,t}$ is uniquely determined by constraints (\ref{u}).
Accordingly, $\textbf{P3}$ can be transformed into
\begin{subequations}
	\begin{align}
		\textbf{P4}~\min_{\mathcal{A}_t,\mathcal{Z}_t,\mathcal{U}_t, \Upsilon_t} &\Upsilon_t+\alpha_2 \sum_{k=1}^{K} (\frac{P_kC_{k}u_{k,t}}{R_{k,t}} + a_{k,t}d_k\epsilon s_k f_{k,t}^2E),\nonumber
	\end{align} 
	\begin{align}
		s.t.~(\ref{c11}),(\ref{c51}),(\ref{z}),(\ref{zetaz}),(\ref{u}).\nonumber
	\end{align}
\end{subequations}
Since $a_{k,t}$ is a binary variable in (\ref{c11}), $\textbf{P4}$ remains a mixed-integer problem. Therefore, we express constraint (\ref{c11}) as the intersection of the following constraints,
\begin{align}\label{asum}
	\sum_{k=1}^{K}(a_{k,t}-a_{k,t}^2)\leq0,
\end{align}
\begin{align}\label{ap}
	0\leq a_{k,t}\leq1, k\in\mathcal{K}'.
\end{align}
In this way, $\textbf{P4}$ can be reformulated into
\begin{subequations}
	\begin{align}
		\textbf{P5}~\min_{\mathcal{A}_t,\mathcal{Z}_t,\mathcal{U}_t,\Upsilon_t} &\Upsilon_t+\alpha_2 \sum_{k=1}^{K} (\frac{P_kC_{k}u_{k,t}}{R_{k,t}} + a_{k,t}d_k\epsilon s_k f_{k,t}^2E),\nonumber
	\end{align} 
	\begin{align}
		s.t.~(\ref{c51}),(\ref{z}),(\ref{zetaz}),(\ref{u}),(\ref{asum}), (\ref{ap}).\nonumber
	\end{align}
\end{subequations}
The constraint in (\ref{asum}) can be expressed as the difference of two convex functions. Since the remaining constraints and the objective function in $\textbf{P5}$ are convex, we convert $\textbf{P5}$ into $\textbf{P6}$ by introducing a penalty factor $\rho$, given by
\begin{subequations}
	\begin{align}
		\textbf{P6}~\min_{\mathcal{A}_t,\mathcal{Z}_t, \mathcal{U}_t,\Upsilon_t} &\Upsilon_t +\alpha_2 \sum_{k=1}^{K} (\frac{P_kC_{k}u_{k,t}}{R_{k,t}} + a_{k,t}d_k\epsilon s_k f_{k,t}^2E)+\nonumber\\
		&\rho\sum_{k=1}^{K}(a_{k,t}-a_{k,t}^2),\nonumber
	\end{align}
	\begin{align}
		s.t.~(\ref{c51}),(\ref{z}),(\ref{zetaz}),(\ref{u}),(\ref{ap}).\nonumber
	\end{align}
\end{subequations}
It has been demonstrated that $\textbf{P5}$ is equivalent to $\textbf{P6}$ when $\rho$ is sufficiently large \cite{gotoh2018dc}. Moreover, $\textbf{P6}$ can be expressed as a difference of convex functions, given by $\Upsilon_t+\alpha_2\sum_{k=1}^{K}\frac{P_kC_ku_{k,t}}{R_{k,t}}+a_{k,t}d_k\epsilon s_kf_{k,t}^2E+\rho\sum_{k=1}^{K}a_{k,t}$ and $\sum_{k=1}^{K}\rho a_{k,t}^2$. In general, we can solve $\textbf{P6}$ by optimizing the sub-problem $\textbf{P6}[\tau']$ in each iteration $\tau'$ \cite{tono2017efficient},
\begin{subequations}
	\begin{align}
		\textbf{P6}[\tau']~\min_{\mathcal{A}_t,\mathcal{Z}_t,\mathcal{U}_t,\Upsilon_t} &\Upsilon_t +\alpha_2\sum_{k=1}^{K} (\frac{P_kC_{k}u_{k,t}}{R_{k,t}} + d_k\epsilon s_k f_{k,t}^2E)\nonumber\\
		&+\rho\sum_{k=1}^{K}a_{k,t}-\rho\sum_{k=1}^{K}{a_{k,t}^{[\tau'-1]}}^2\nonumber\\
		&-2\rho\sum_{k=1}^{K}a_{k,t}^{[\tau'-1]}(a_{k,t}-a_{k,t}^{[\tau'-1]}),\nonumber
	\end{align}
	\begin{align}
		s.t.~(\ref{c51}),(\ref{z}),(\ref{zetaz}),(\ref{u}),(\ref{ap}).\nonumber
	\end{align}
\end{subequations}
The whole method to solve $\textbf{P6}$ using the DC algorithm is summarized in Algorithm \ref{dc}.
\begin{algorithm}[t]
	\caption{DC-based Algorithm to $\textbf{P6}$}
	\label{dc}
	\begin{algorithmic}[1]
		\renewcommand{\algorithmicrequire}{\textbf{Input:}}
		\renewcommand{\algorithmicensure}{\textbf{Output:}}
		\REQUIRE A feasible solution 
		$\{\mathcal{A}_t,\mathcal{Z}_t,\mathcal{U}_t,\Upsilon_t\}$ for $\textbf{P6}$.
		\STATE Set $\tau' \gets 0$
		
		\WHILE{$\tau' \leq \tau_{\max}^{\mathrm{DC}}$}
		\STATE $\tau' \gets \tau' + 1$
		\STATE Solve problem $\textbf{P6}[\tau'$] by interior point method to obtain 
		$\mathcal{A}_t$, $\mathcal{Z}_t$, $\mathcal{U}_t$, $\Upsilon_t$, and set 
		$a_{k,t}^{[\tau']} = a_{k,t}$
		\ENDWHILE
		
		\IF{$a_{k,t} > 0.5$}
		\STATE Set $b_{k,t} = \frac{1}{u_{k,t}}$
		\ENDIF
		
		\ENSURE
		$\mathcal{A}_t$, $\mathcal{B}_t$, $\Upsilon_t$.
		
	\end{algorithmic}
\end{algorithm}
\begin{algorithm}[t]
	\renewcommand{\algorithmicrequire}{\textbf{Input:}}  
	\renewcommand{\algorithmicensure}{\textbf{Output:}}           
	\caption{Client Selection and Resource Allocation (CSRA) Algorithm}
	\label{joint}
	\begin{algorithmic}[1]
		
		\REQUIRE
		$\alpha_1$, $\alpha_2$, $B$, $\{f_{k,t}^{\max}\}$, $e_1^{\max}$, $e_2^{\max}$, $\{\mathcal{D}_k\}$

		\STATE Select the clients satisfies $D_{kl}(\hat{p_g} \| \hat{p}_k) \leq e_1^{max}, \forall k \in \mathcal{K}$.
		\STATE Employ \textbf{Algorithm}~\ref{dc} to obtain $\mathcal{A}_t$, $\mathcal{B}_t$, and $\Upsilon_t$.
		
		\STATE Set $\tau' \gets 0$.
		\WHILE {$\tau' \leq \tau_{max}^{SG}$}
		\STATE Substitute the dual variables $\gamma_k^{[\tau']}$ and $\beta_k^{[\tau']}$ into (\ref{fopt}) to obtain $f_{k,t}$.
		\STATE Update the dual variables $\gamma_k^{[\tau'+1]}$ and $\beta_k^{[\tau'+1]}$ by (\ref{dualupsilon}) and (\ref{dual_iota}).
		\STATE $\tau' \gets \tau' + 1$.
		\ENDWHILE

		\ENSURE
		$\mathcal{A}_t$, $\mathcal{B}_t$, $\mathcal{F}_t$
		
	\end{algorithmic}
\end{algorithm}

Given the obtained values of $\mathcal{A}_t$, $\mathcal{B}_t$, $\Upsilon_t$ from $\textbf{P6}$, we further optimize the computation capacity. By substituting $\Upsilon_t$ into $\Upsilon_t'$ in $\textbf{P1}$, the objective function in step (b) is given by 
\begin{subequations}
	\begin{align}
		\textbf{P7}~\min_{\mathcal{F}_t,\Upsilon_t'} &~ \Upsilon_t'+\alpha_2 \sum_{k=1}^{K} (\frac{P_kC_{k}a_{k,t}}{b_{k,t}R_{k,t}} + a_{k,t}d_k\epsilon s_k f_{k,t}^2E),\nonumber
	\end{align} 
	\begin{align}
		s.t.~(\ref{c31}),\nonumber
	\end{align}
	\begin{align}\label{zeta31}
		\alpha_1 {\frac{a_{k,t}C_{k}}{b_{k,t}R_{k,t}}}+
		\alpha_1{ \frac{Ea_{k,t}s_kd_{k}}{f_{k,t}}}\leq\Upsilon_t', \forall k \in \mathcal{K}'.
	\end{align}
\end{subequations}
$\textbf{P7}$ contains a quadratic objective function with convex constraints.  
By introducing non-negative dual variables $\boldsymbol{\gamma}$ and $\boldsymbol{\beta}$, we can derive the Lagrangian function $\mathcal{L}(\mathcal{F}_t, \boldsymbol{\gamma}, \boldsymbol{\beta})$ of $\textbf{P7}$ as
\begin{align}
	&\mathcal{L}(\mathcal{F}_t, \boldsymbol{\gamma}, \boldsymbol{\beta})\nonumber\\ 
    &= \Upsilon_t'+\alpha_2 \sum_{k=1}^{K} a_{k,t} (\frac{P_kC_{k}}{b_{k,t}R_{k,t}} + d_k\epsilon s_k f_{k,t}^2E)+
	\sum_{k=1}^{K}\gamma_k(f_{k,t}\nonumber\\
    &-f_{k,t}^{max})
    +\sum_{k=1}^{K}\beta_k(\alpha_1 {\frac{a_{k,t}C_{k}}{b_{k,t}R_{k,t}}}+\alpha_1{ \frac{Ea_{k,t}s_kd_{k}}{f_{k,t}}}-\Upsilon'_t).
\end{align}
Due to the convexity of $\textbf{P7}$, the optimal solution satisfies the Karush
Kuhn-Tucker (KKT) conditions.
By taking the first order derivatives of $f_{k,t}$, we obtain
\begin{align} \label{fk}
	\frac{\partial \mathcal{L}(\mathcal{F}, \boldsymbol{\gamma}, \boldsymbol{\beta})}{\partial f_{k,t}}=2\alpha_2d_k\epsilon s_kf_{k,t}E+\gamma_k-\frac{E\beta_k\alpha_1s_kd_k}{f_{k,t}^2}.
\end{align}
By setting (\ref{fk}) to zero, we obtain the optimal solution of $f_{k,t}$ with respect to $\gamma_k$ and $\beta_k$ as
\begin{equation}
	\label{fopt}
	f_{k,t}^* =
	\begin{cases}
		(-q-q'[\sqrt[3]{X_1}+\sqrt[3]{X_2}])^{\frac{1}{3}}, & \Delta\textgreater0, [\sqrt[3]{X_1}+\sqrt[3]{X_2}]\textless\frac{3}{2};\\
		\frac{\gamma_k}{6\alpha_2d_k\epsilon s_k E}, & \Delta = 0; \\
		A\cos\frac{\arccos C}{3}-G, & \Delta <0; \\
		\text{no solution}, & \text{otherwise,}
	\end{cases}
\end{equation}
where $\Delta=\frac{(\beta\alpha_1)^2}{16(\alpha_2\epsilon)^2}-\frac{\gamma_k^3\beta_k\alpha_1}{432(\alpha_2\epsilon)(\alpha_2d_k\epsilon s_kE)^3}$, $q=\frac{2\gamma_k^3-108(\alpha_2d_k\epsilon s_kE)^2(E\beta_k\alpha_1s_kd_k)}{216(\alpha_2d_k\epsilon s_kE)^3}$, $q'=\frac{\gamma_k^3}{162(\alpha_2d_k\epsilon s_kE)^3}-\frac{\beta_k\alpha_1}{3\alpha_2\epsilon}$,  $X_1=\frac{\beta_k\alpha_1}{2\alpha_2\epsilon}-\frac{\gamma_k^3}{108(\alpha_2d_k\epsilon s_kE)^3}+\sqrt{\Delta}$, $X_2=\frac{\beta_k\alpha_1}{2\alpha_2\epsilon}-\frac{\gamma_k^3}{108(\alpha_2d_k\epsilon s_kE)^3}-\sqrt{\Delta}$, $A=\frac{\gamma_k}{3\alpha_2d_k\epsilon s_kE}$, $G=\frac{\gamma_k}{6\alpha_2d_k\sigma s_k E}$, $C=\frac{54\beta_k\alpha_1\alpha_2^2\sigma^2d_k^3s_k^3E^3}{\gamma_k^3}-1$.
Since $\textbf{P7}$ is a convex optimization problem and satisfies the Slater condition, the duality gap between the primal and dual problem is zero. Thus, the dual sub-gradient method can be employed to compute the dual variables $\gamma_k$, $\beta_k$ in a distributed manner \cite{9149343}. The updates of $\gamma_k$, $\beta_k$ at the $\tau'+1$-th iteration are given by
\begin{align}\label{dualupsilon}
	\gamma_k^{[\tau'+1]}=[ \gamma_k^{[\tau']} +\upsilon_k^{[\tau']}(f_{k,t}^*-f_{k,t}^{max})]^+,
\end{align}
\begin{align}\label{dual_iota}
	\beta_k^{[\tau'+1]}=[ \beta_k^{[\tau']} +\iota_k^{[\tau']}(\alpha_1{ \frac{Es_kd_{k}}{f_{k,t}^*}}-\Upsilon_t')]^+,
\end{align}
where $[x]^+=\max\{0,x\}$, $\upsilon_k^{[\tau']}$ and $\iota_k^{[\tau']}$ denote the updating step sizes of the $\tau'$-th iteration.

Combing the solutions to $\textbf{P6}$ and $\textbf{P7}$, we present the joint client selection and resource allocation method (CSRA) in Algorithm \ref{joint}.
\subsection{Complexity Analysis}
The complexity of CSRA comes from two parts. The first part is the complexity of DC algorithm. There are a total optimization $3K+1$ variables and $6K+2$ convex constraints. The complexity of DC algorithm is $\mathcal{O}((3K+1)^3(6K+2))=\mathcal{O}(K^4)$. The second part of complexity comes from sub-gradient algorithm, we first compute the primal closed-form responses 
$f_{k,t}^*$  with respect to the dual variables obtained in equation (\ref{fopt}), which costs $\mathcal{O}(K)$. Then we update $2K$ dual variables, which also costs $\mathcal{O}(K)$. 
Hence, the complexity of sub-gradient method is $\mathcal{O}((K+2K))=\mathcal{O}(K)$. Since the overall complexity cost is dominated by solving DC algorithm, the total computational complexity of CSRA is $\mathcal{O}(K^4)$.

\section{Simulation and Performance Evaluations} 
\begin{table}[tbp]
	\vspace{-8pt}
	\centering
	\caption{\label{tab:test} Simulation Parameter Setting}
	\setlength\tabcolsep{15pt}{
		\begin{tabular}{lcl}
			\toprule
			Parameter & Value \\
			\midrule
			The number of devices $K$ & 80\\
			Distance from device to server $dis_k$ & U[200,250]m\\
			Power spectrum density of noise $N_0$ & -174dBm/Hz  \\
			Average small-scale channel gain & 1\\
			Total bandwidth $B^{up}$ & 2 MHz  \\
			Transmit power $P^{up}_k$ & [20,33] dBm\\
			Maximum computation capacity $f^{max}_k$ & [2,5] MHz\\
			Capacitance coefficient $\epsilon$ & 10\textsuperscript{-27} \\
			Processing density $s_k$ & [1,10] cycles/bit\\
			Path loss exponent $\gamma$ & 2.7\\
			Carrier frequency $f_c$ & 2.4 GHz\\
			Maximum acceptable energy $E$ & [2,9] J\\
			System parameter $\mu$ & 1.7$\times$10\textsuperscript{-6}\\	
			\bottomrule
	\end{tabular}}\label{parameter}
\end{table}\label{simulation}
In this section, we evaluate the performance of CSRA in terms of test accuracy, learning latency, and energy consumption. The experimental setup and simulation results are provided in the following.

\subsection{Experiment Setup}
The performance of CSRA is evaluated on Fashion-MNIST\cite{xiao2017fashion} and CIFAR-10 \cite{krizhevsky2009learning} datasets.
We adopt a convolutional neural network consisting of two convolutional layers with 6 and 16 channels, respectively, followed by three fully connected layers. ReLU activations are applied after each layer, max-pooling with a stride of 2 is used for downsampling, and a final softmax layer performs classification. The datasets are partitioned using a hybrid data partitioning strategy \cite{zhang2021client}. Specifically, $10\%$ of the clients receive an equal number of samples per category, while the remaining clients follow a Dirichlet-based non-IID split with concentration parameter $0.5$. 
We consider an urban environment with $80$ clients
randomly distributed within the coverage range of a base station equipped with an edge server. The distance between clients and the server ranges from 200 m to 250 m. Both large-scale and small-scale fading are considered in the channel model. 25\% clients are selected for training with the rest for testing. 
The relevant simulation parameters are listed in Table \ref{parameter}.

To demonstrate the effectiveness of CSRA in terms of model performance, we compare it with four baseline methods: (a) the SAE algorithm \cite{9237168}, which selects clients based on instantaneous wireless channel conditions. (b) FedAvg \cite{mcmahan2017communication}, which randomly selects the same number of clients as CSRA in each round for federated learning. (c) Pow \cite{cho2022towards}, which selects clients with the maximum loss. (d) CFL \cite{10192896}, which clusters clients according to gradient similarity and performs FL within each cluster.
We further evaluate the effectiveness of CSRA in terms of training costs by comparing it with the following baseline methods: (a) CS-Random, which selects clients by Algorithm \ref{dc}, while bandwidth and computation capacity are allocated randomly; (b) GA-Random, which selects clients by the Genetic Algorithm (GA), with bandwidth and computation capacity allocated randomly; (c) CS-Greedy, which selects clients by Algorithm \ref{dc} and allocates bandwidth and computation capacity via a greedy algorithm; (d) GA-Greedy, which selects clients by GA and allocates bandwidth and computation capacity by greedy algorithm.
For the GA-based baselines, the client selection is guided by a fitness function defined as
\begin{equation}
	\label{fitness}
	fitness = M(c_1^2+c_2^2),
\end{equation}
where $c_1=\max(0,\sum_{k=1}^{K} (a_{k,t}D_{kl}({\hat{p}_g} \| \hat{p}_k) - e_1^{max}))$, and $c_2=\max(0,\sum_{k=1}^{K}a_{k,t}d_k - {1}/{e_2^{max}})$, where $M$ is a penalty factor.
For the greedy-based baselines, bandwidth and capacity are allocated as follows.
Step (1) calculate the marginal gain of devices by $marginal\_gain(k)=-\alpha_1\frac{a_{k,t}C_k}{b_{k,t}^2R_{k,t}}-\alpha_2\frac{P_kC_k}{b^2_{k,t}R_{k,t}}$. Step (2) Choose a device with the maximum marginal gain to allocate bandwidth $b_k^*=b_k^*+\Delta b$, where $\Delta b$ is the minimum bandwidth allocation unit. Step (3) Update the computation capacity for each client as $f_{k,t}^*=\max{(0,\min(f_k^{max},(\frac{\alpha_1a_{k,t}}{2\alpha_2d_k\sigma s_kE})^{\frac{1}{3}})}$.
\subsection{Model Performance Evaluation}
\subsubsection{Model Performance}
\begin{figure}[tbp]
	\centering
	\subfloat[Accuracy on Fashion-MNIST\label{acc fmnist}]{
		\includegraphics[scale=0.45]{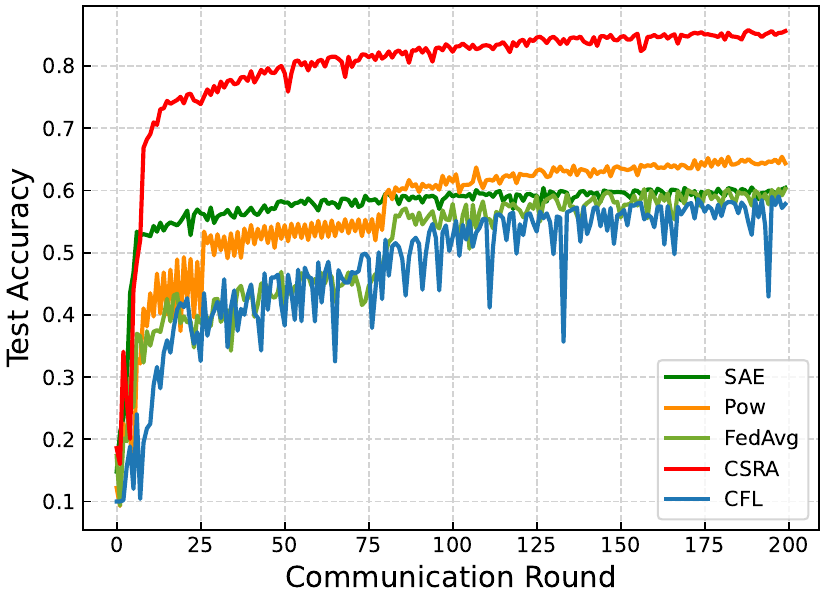}}
	\\
	\subfloat[Accuracy on Cifar-10\label{acc cifar10}]{
		\includegraphics[scale=0.45]{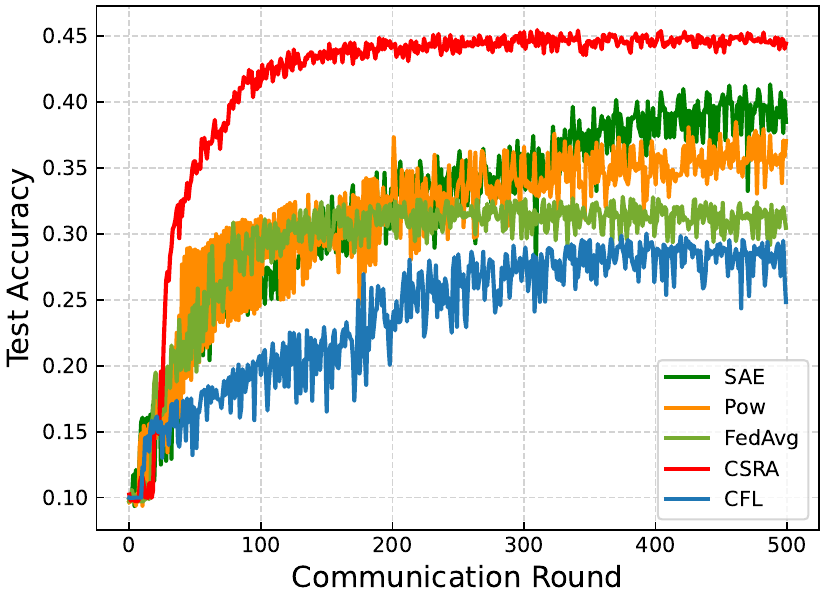}}
	\caption{Comparison of test accuracy between CSLP and other baseline methods on (a) Fashion-MNIST and (b) CIFAR-10.}
	\label{acc}
\end{figure}
To evaluate the performance improvement achieved by CSRA, we conduct experiments on the Fashion MNIST and CIFAR-10 datasets. All methods use the same initial model, with learning rates of 0.005 and 10 local epochs. As illustrated in Figs. \ref{acc fmnist} and \ref{acc cifar10}, CSRA outperforms all other methods in terms of test accuracy. 
The loss-based Pow method provides limited accuracy improvement compared with the naive FedAvg approach. SAE, which selects clients based on channel quality, cannot consistently guarantee accuracy across different datasets. In particular, on the Fashion-MNIST dataset, its performance is comparable to that of FedAvg. 
The CFL method suffers from unstable performance due to frequent changes in cluster assignments. 
The improvement of CSRA is attributed to 
it selects clients whose data distributions align with the global distribution while ensuring a sufficient number of data samples. 

\subsubsection{Effect of Distribution Divergence}
\begin{figure}[tbp]
\vspace{-8pt}
	\centering
	\subfloat[Accuracy on Fashion-MNIST\label{fmnist kl}]{
		\includegraphics[scale=0.45]{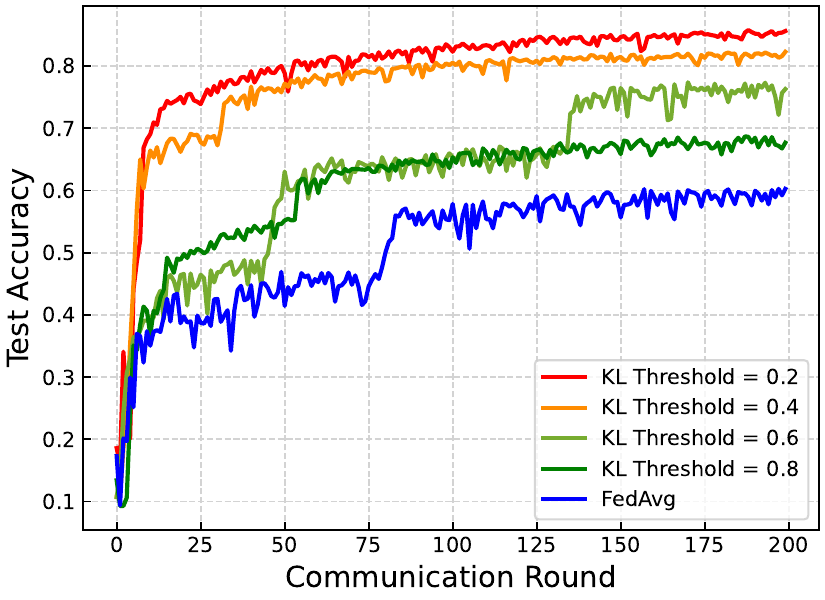}}
	\\
	\subfloat[Accuracy on Cifar-10 \label{cifar10 kl}]{
		\includegraphics[scale=0.45]{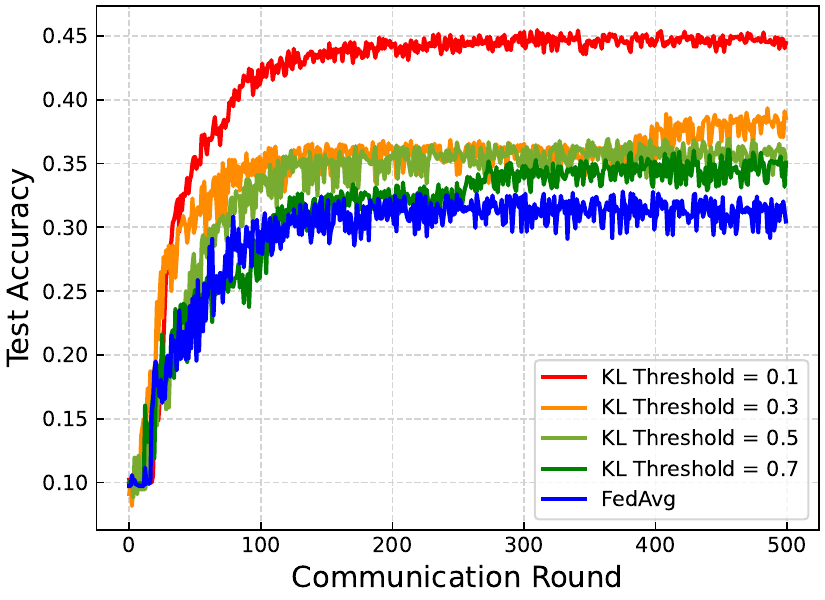}}
	\caption{Effect of the distribution divergence (KL) threshold on test accuracy for (a) Fashion-MNIST and (b) CIFAR-10.}
	\label{kl}
\end{figure}

Figs.~\ref{fmnist kl} and \ref{cifar10 kl} illustrate the test accuracy on the Fashion-MNIST and CIFAR-10 datasets under different KL divergence thresholds $e_1^{max}$, with the data size budget set to $e_2^{max}=2000$. It can be observed that a larger $e_1^{max}$ results in poorer performance. This is because a larger $e_1^{max}$ represents a larger data distribution divergence among clients. When a small $e_1^{max}$ is chosen, CSRA tends to select devices whose local data distributions closely match the global distribution. However, as the constraint on distribution divergence is relaxed, the test accuracy gradually approaches that of the baseline FedAvg with random client selection. 
Hence, enforcing an appropriate divergence threshold remains essential to preserve stable and robust test performance.
\subsubsection{Effect of Data Size Budget}
\begin{figure}[tbp]
	\centering
	\subfloat[Accuracy on Fashion-MNIST\label{fmnist datasize}]{
		\includegraphics[scale=0.45]{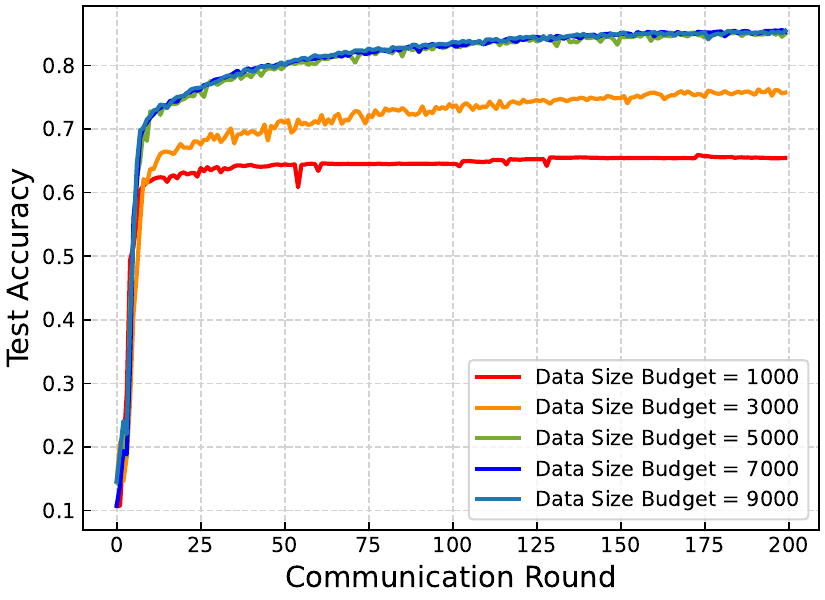}}
	\\
	\subfloat[Accuracy on Cifar-10\label{cifar10 datasize}]{
		\includegraphics[scale=0.45]{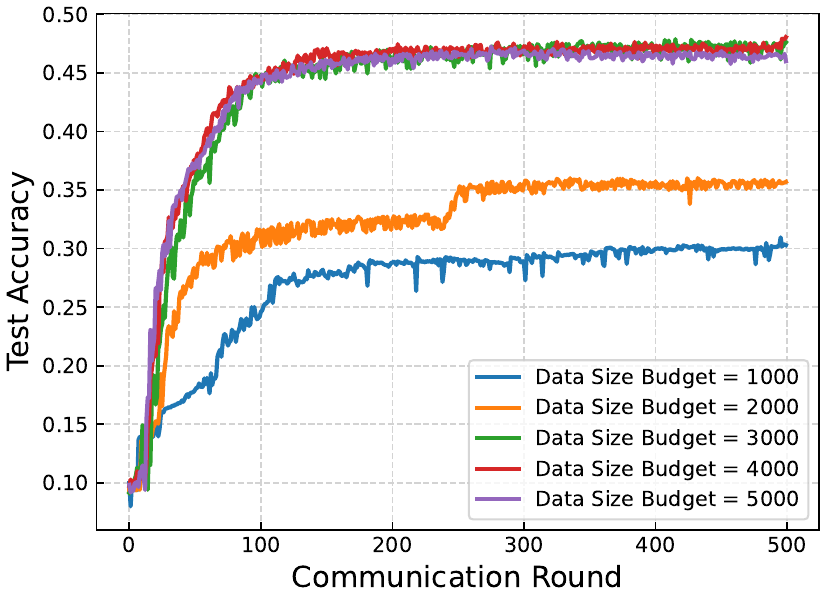}}
	
	\caption{Effect of the data size budget on test accuracy for (a) Fashion-MNIST and (b) CIFAR-10.}
	\label{acc_size}
\end{figure}
%
The test accuracy on Fashion-MNIST and CIFAR-10 under different data size budget $e_2^{max}$ is shown in Figs. \ref{fmnist datasize} and \ref{cifar10 datasize}, with  the KL divergence constraint $e_1^{max}$ fixed at $0.2$. 
Initially, the test accuracy of CSRA improves as the data size budget increases. This improvement is attributed to the larger dataset amount of data used in each training round, which allows the model to learn from a broader range of samples and categories.
However, as the data size budget continues to increase, the improvement in model performance becomes less pronounced. This is because the additional data provides increasingly repetitive knowledge, offering imited further contribution to the model performance.
\subsection{Training Cost Evaluation}
\subsubsection{Learning Latency and Energy Costs}

To compare the energy consumption and learning latency of CSRA with baseline methods, we set $\alpha_1 = \alpha_2 = 1$ for all schemes. Figs. \ref{energy round} and \ref{latency round} show the cumulative energy consumption and learning latency over training rounds, respectively. The Random-based method results in higher accumulated latency, mainly due to the randomness in resource allocation. In contrast, the Greedy-based method can effectively reduce learning latency compared with the Random-based methods. Although Greedy-based methods incur slightly higher energy consumption compared with Random-based methods, this drawback can be mitigated by appropriately tuning $\alpha_1$ and $\alpha_2$. Among all methods, CSRA achieves the lowest energy consumption and latency because it jointly optimizes client selection, bandwidth, and computation capacity allocation rather than simply allocating bandwidth to clients with the highest marginal gain or allocating resources randomly. These results highlight the effectiveness of CSRA in reducing latency and energy consumption through efficient resource allocation and client selection. 

\begin{figure}[tbp]
	\centering
	
	\subfloat[Energy Consumption \label{energy round}]{
		\includegraphics[scale=0.45]{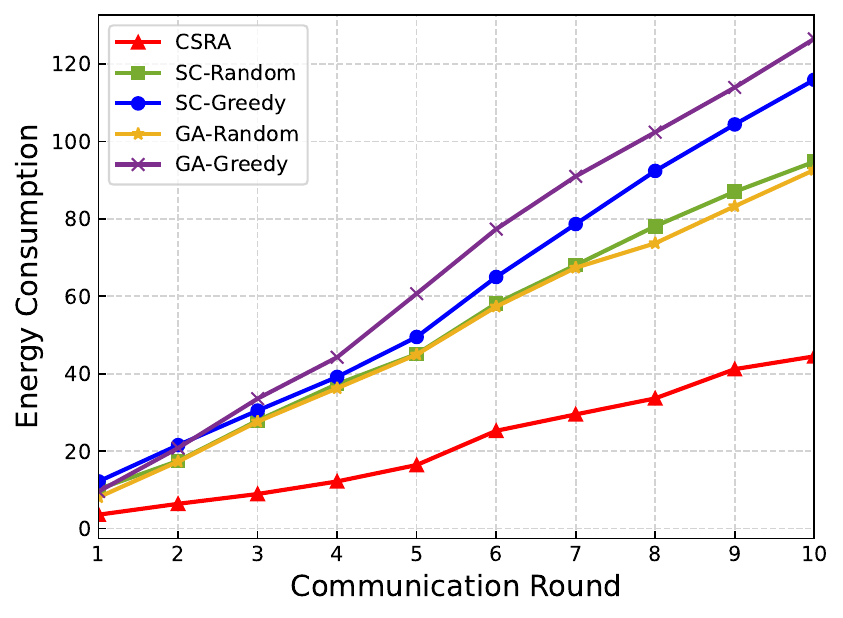}}
	\\
	\subfloat[Learning Latency \label{latency round}]{
		\includegraphics[scale=0.45]{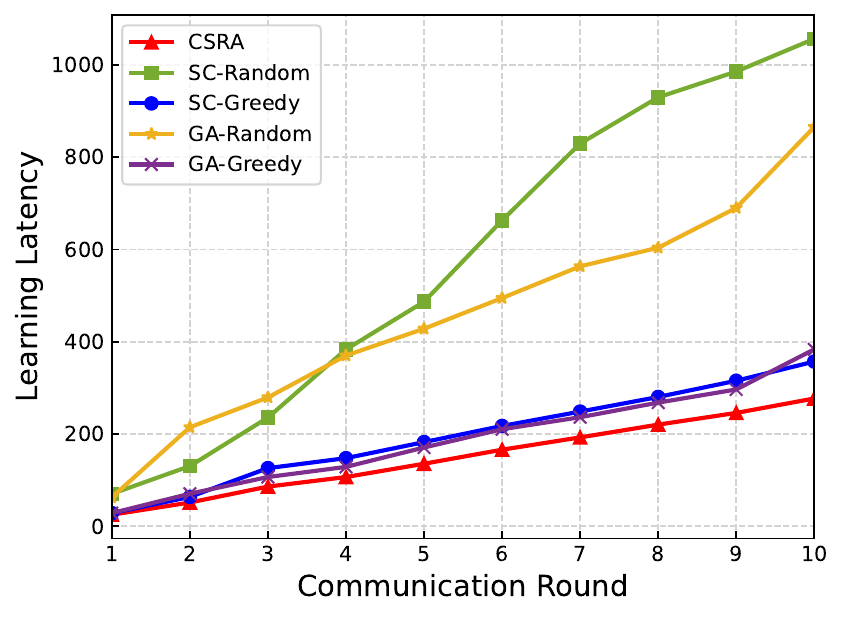}}
	
	\caption{Comparison of (a) Energy consumption and (b) Learning latency between CSRA and other methods.}
	\label{fig:main}
\end{figure}

\subsubsection{Effect of Bandwidth}
\begin{figure}[tbp]
	\centering
	\vspace{7pt}
	\subfloat[Energy Consumption\label{bandwidth energy}]{
		\includegraphics[scale=0.45]{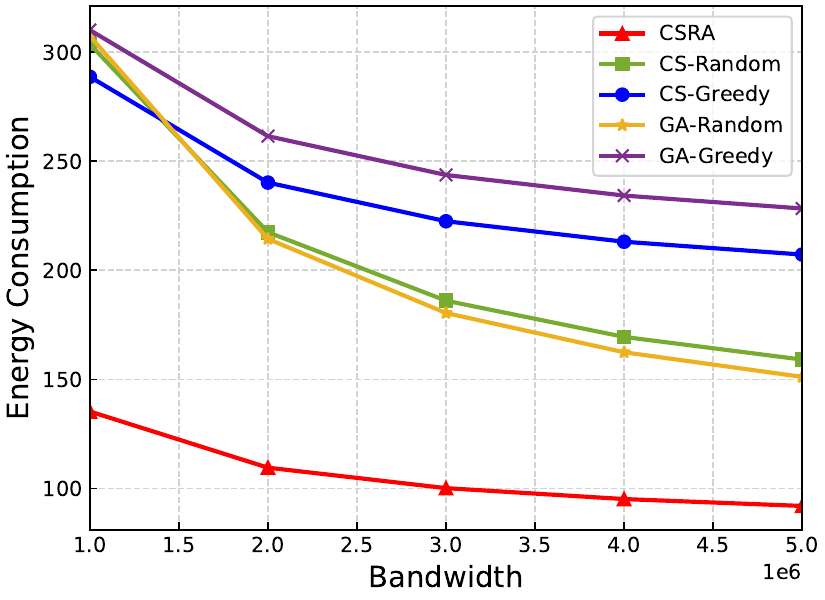}}
	\\
	\subfloat[Learning Latency\label{bandwidth latency}]{
		\includegraphics[scale=0.45]{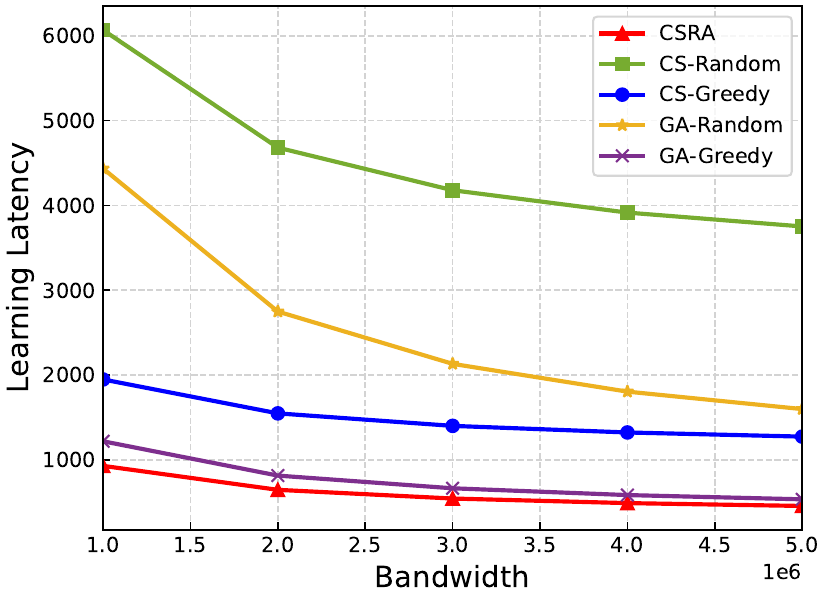}}
	
	\caption{Effect of bandwidth on (a) Energy consumption and (b) Learning latency.}
	\label{fig:bandwidth}
\end{figure}
Figs. \ref{bandwidth energy} and \ref{bandwidth latency} illustrate the energy consumption and learning latency under different total bandwidth values. It can be observed that both learning latency and energy consumption decrease as the total bandwidth increases. This is because allocating more bandwidth to clients reduces transmission latency.
Additionally, as bandwidth continues to increase, its impact on learning latency and energy consumption diminishes. This is because, once transmission latency decreases to a certain threshold, transmission latency and energy are no longer the dominant influencing factors of the total training costs. Notably, CSRA achieves greater reductions in both energy consumption and latency, particularly when the total bandwidth is limited. This demonstrates the effectiveness of CSRA in optimizing resource allocation under constrained bandwidth conditions.
\subsubsection{Effect of Channel Gain}
\begin{figure}[tbp]
	\centering
	\subfloat[Energy Consumption\label{channel energy}]{
		\includegraphics[scale=0.45]{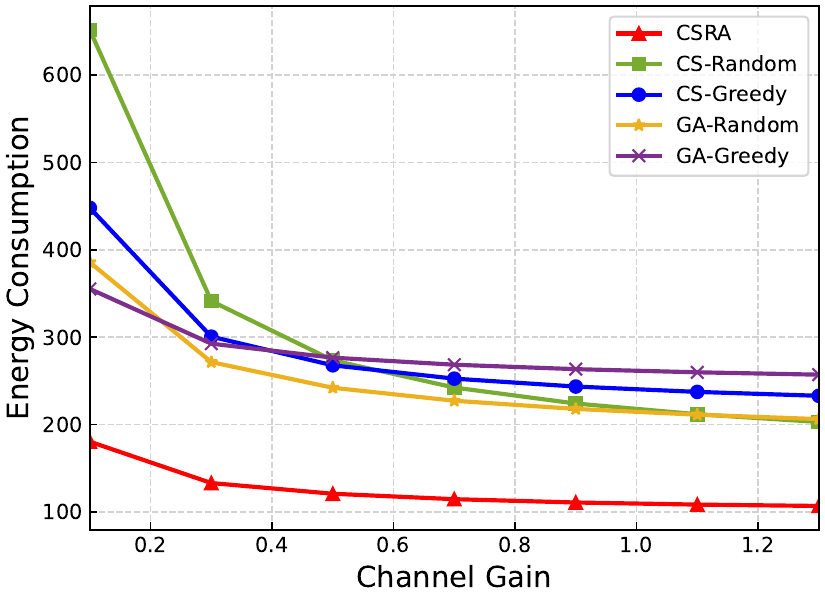}}
	\\
	\subfloat[Learning Latency\label{channel latency}]{
		\includegraphics[scale=0.45]{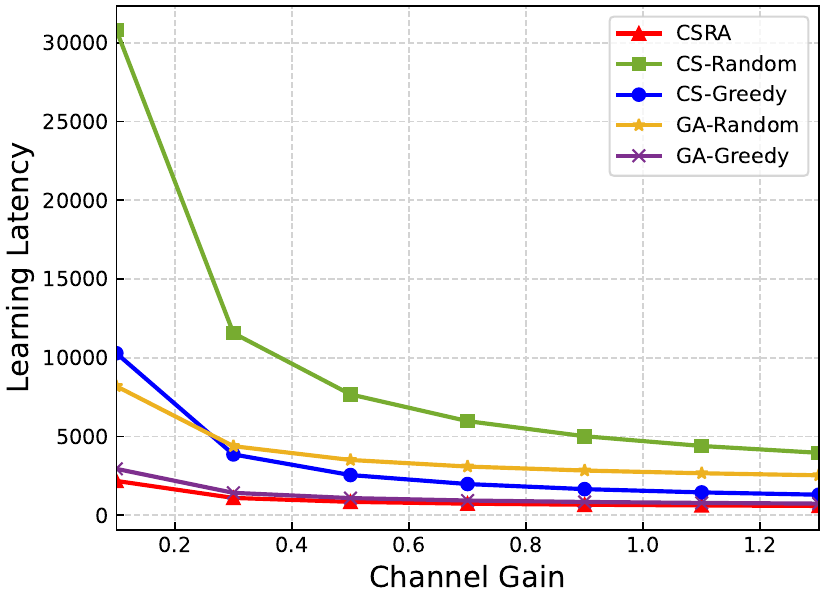}}
	
	\caption{Effect of channel gain on (a) Energy consumption and (b) Learning Latency.}
	\label{fig:channel}
\end{figure}
To demonstrate the effectiveness of CSRA under different channel conditions, we vary the channel gain from 0.1 to 1.3. As shown in Figs. \ref{channel energy} and \ref{channel latency}, CSRA consistently achieves lower energy and latency. This demonstrates that CSRA is more robust to channel fluctuations and can maintain stable performance even when channel quality degrades. In contrast, the baseline methods are more sensitive to channel variations. These results indicate that jointly optimizing client selection and resource allocation enables CSRA to mitigate the adverse impact of channel degradation effectively.

\subsubsection{Effect of Data Size Budget}
Figs. \ref{size energy} and \ref{size latency} present the impact of the data size budget on energy consumption and learning latency. For all methods, both latency and energy increase steadily with larger data size budget, while CSRA exhibits slight fluctuations compared with the baseline methods. This is mainly because increasing the data size budget leads to higher computation workloads, resulting in increased computation costs.
Overall, CSRA outperforms all baseline methods in both learning latency and energy consumption. However, as shown in Fig. \ref{acc_size}, the improvement in test accuracy becomes increasingly marginal with a larger data size budget. Hence, using excessively large training data samples per round is not recommended, as it offers limited accuracy gains but significantly increases both latency and energy consumption.

\begin{figure}[tbp]
	\centering
	\subfloat[Energy Consumption \label{size energy}]{
		\includegraphics[scale=0.45]{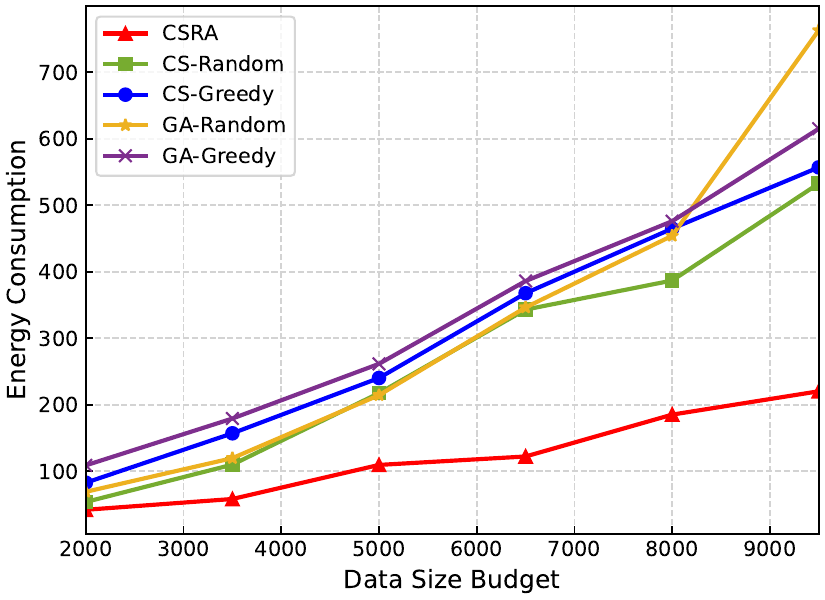}}
	\\
	\subfloat[Learning Latency \label{size latency}]{
		\includegraphics[scale=0.45]{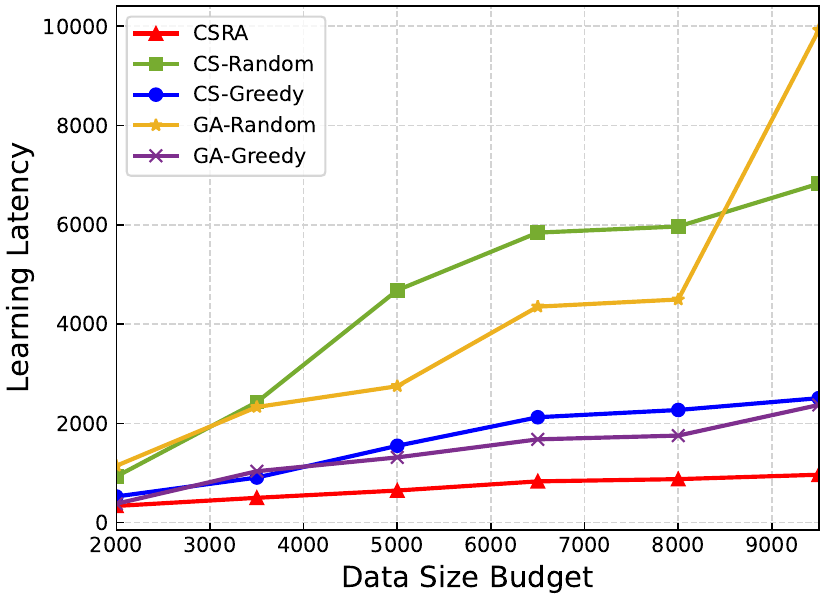}}
	
	\caption{Effect of data size budget on  (a) Energy consumption and (b) Learning Latency.}
	\label{size}
	
\end{figure}
\subsubsection{Effect of Trade-off Parameters}	
\begin{figure}[tbp]
\vspace{4pt}
    \hspace*{0.05\textwidth}
	\centering
	\includegraphics[scale=0.45]{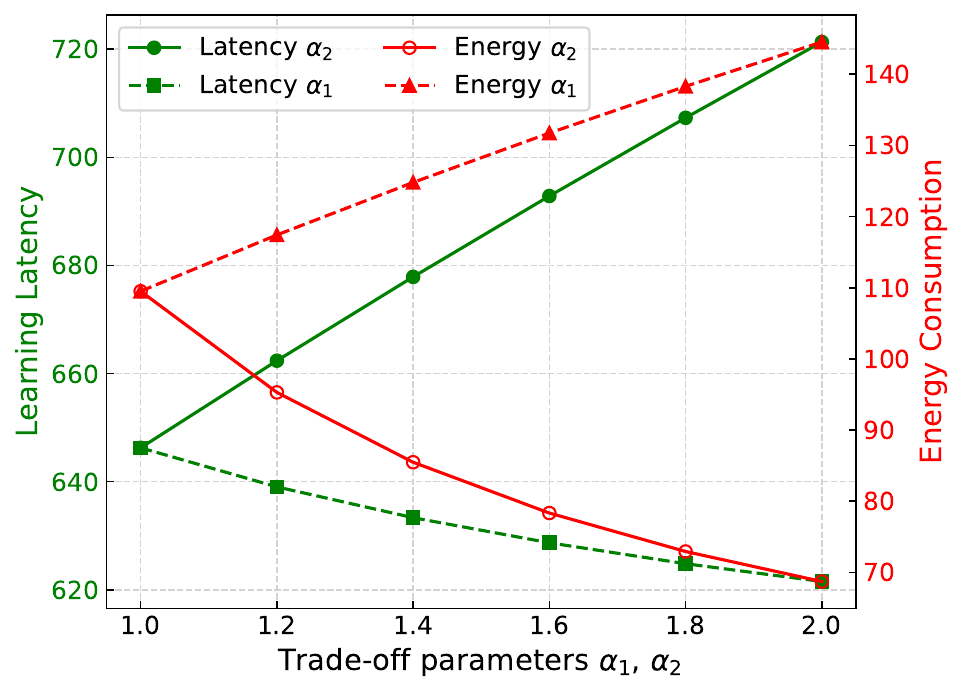}
	\caption{Effect of trade-off parameters $\alpha_1$ and $\alpha_2$.} \label{param}
\end{figure}

The impact of the trade-off parameters $\alpha_1$ and $\alpha_2$ on energy consumption and learning latency is illustrated in Fig. \ref{param}, where one parameter varies from $1.0$ to $2.0$ while keeping the other fixed at $1.0$. As shown, increasing $\alpha_1$ leads to higher energy consumption but lower latency. This is because $\alpha_1$ prioritizes latency in the objective function, and a higher value of $\alpha_1$ makes latency the dominant factor, prompting CSRA to optimize it accordingly. In contrast, the parameter $\alpha_2$ primarily controls energy consumption. Increasing $\alpha_2$ exhibits the opposite trend, leading to reduced energy consumption but increased learning latency. These results demonstrate that CSRA effectively balances and controls both latency and energy consumption, achieving the objective of joint optimization of latency and energy.

\section{conclusions} \label{conclusions}
In this paper, we investigated the effect of data heterogeneity and training cost optimization in FL. We proposed a joint client selection and resource allocation approach (CSRA) to optimize energy consumption and learning latency while maintaining high model accuracy. 
The following conclusions were drawn from the simulations: 1) Selecting clients with low distribution divergence is an effective strategy to improve model performance. 
2) Excessively increasing the data size budget per round can only offer limited accuracy improvement, while significantly increasing latency and energy consumption. 
3) CSRA effectively improves test accuracy while simultaneously reducing and flexibly controlling learning latency and energy consumption through the trade-off parameters.

\bibliographystyle{IEEEtran}
\bibliography{References1}

\vfill

\end{document}